\documentclass[runningheads]{llncs}
\pdfoutput=1

\usepackage[T1]{fontenc}
\usepackage{graphicx}

\usepackage[colorlinks=true, allcolors=blue]{hyperref}
\usepackage{color}

\usepackage{makecell}

\usepackage{microtype}
\usepackage{cite}
\usepackage{amsmath}
\usepackage{amssymb}
\usepackage{booktabs}

\usepackage{tikz}
\usetikzlibrary{automata, positioning, fit}
\usetikzlibrary[positioning,calc,arrows.meta,automata]

\usepackage{pgfplots}
\pgfplotsset{compat=1.18}
\usepgfplotslibrary{fillbetween}

\usepackage[capitalise]{cleveref}

\makeatletter
\newsavebox{\@brx}
\newcommand{\llangle}[1][]{\savebox{\@brx}{\(\m@th{#1\langle}\)}%
  \mathopen{\copy\@brx\mkern2mu\kern-0.9\wd\@brx\usebox{\@brx}}}
\newcommand{\rrangle}[1][]{\savebox{\@brx}{\(\m@th{#1\rangle}\)}%
  \mathclose{\copy\@brx\mkern2mu\kern-0.9\wd\@brx\usebox{\@brx}}}
\makeatother

\def\Nat{\bbbn}
\def\Real{\bbbr}

\def\ind{{\mathbf{1}}}
\def\Prob{{\sf P}}
\def\Expect{{\sf E}}
\def\mc{{\bf \Phi}}

\renewcommand{\orcidID}[1]{}

\begin{document}
\title{Quantitative Supermartingale Certificates}

\author{
Alessandro Abate\inst{1}\orcidID{0000-0002-5627-9093} 
\and
Mirco Giacobbe\inst{2}\orcidID{0000-0001-8180-0904} 
\and
Diptarko Roy\inst{2}\orcidID{0009-0003-4306-2076}
}

\institute{University of Oxford, UK  \\\email{alessandro.abate@cs.ox.ac.uk}\and
University of Birmingham, UK\\
\email{\{m.giacobbe,d.s.roy\}@bham.ac.uk}
}

\maketitle              
\begin{abstract}
We introduce a general methodology for quantitative model checking and control synthesis with supermartingale certificates. 
We show that every specification that is invariant to time shifts admits a stochastic invariant that bounds its probability from below;
for systems with general state space, the stochastic invariant bounds this probability as closely as desired; 
for systems with finite state space, it quantifies it exactly. 
Our result enables the extension of every certificate for the almost-sure satisfaction of shift-invariant specifications to its quantitative counterpart, ensuring completeness up to an approximation in the general case and exactness in the finite-state case. 
This generalises and unifies existing supermartingale certificates for quantitative verification and control under 
reachability, safety, reach-avoidance, and stability specifications, as well as asymptotic bounds on accrued costs and rewards.
Furthermore, our result provides the first supermartingale certificate for computing upper and lower bounds on the probability of satisfying $\omega$-regular and linear temporal logic specifications.
We present an algorithm for quantitative $\omega$-regular verification and control synthesis based on our method and demonstrate its practical efficacy on several infinite-state examples.

\end{abstract}

\section{Introduction}
Quantitative model checking for probabilistic systems is the problem of computing the probability that a given stochastic dynamical system or probabilistic program satisfies a specification of intended behaviour. Quantitative control synthesis extends this to the construction of a control policy that maximises or meets a threshold for the probability of satisfying a desired objective within a given stochastic environment.
Computing provable bounds on the probability that a system satisfies its specification is crucial for model checking and control synthesis when neither worst-case nor almost-sure satisfaction can be achieved and failure to comply must be tolerated within acceptable margins. Notable examples include many randomized distributed algorithms and cryptographic protocols, cyber-physical systems and biochemical processes under random parameter and input uncertainty, and machine learning algorithms facing aleatoric uncertainty in their data and epistemic uncertainty in their models. 

Algorithmic technologies for quantitative model checking and control synthesis have been developed extensively for 
probabilistic systems. The standard techniques rely on computing the absorbing components, reduction to linear programming, 
tabular value and policy iteration as well as symbolic algorithms 
based on multi-terminal binary decision diagrams~\cite{DBLP:phd/ethos/Parker03,DBLP:conf/sfm/ForejtKNP11,DBLP:conf/tacas/ForejtKNPQ11,DBLP:conf/cav/KwiatkowskaNP11,DBLP:journals/sttt/HenselJKQV22}. 
This represents the state of the art for systems with a finite state space 
but, falls short for systems with a countably infinite or 
continuous 
state space, 
which is common in probabilistic programs, control systems, and machine learning models.
The automated verification and control of infinite-state probabilistic systems builds upon 
either the construction of finite abstractions---grounded in concurrency theory---or 
the construction of proof certificates---grounded in martingale theory~\cite{DBLP:conf/lics/BluteDEP97,DBLP:journals/iandc/LarsenS91,DBLP:conf/tacas/ChakarovVS16,DBLP:conf/cav/ChakarovS13,jpk}. 

Proof certificates for the analysis of dynamical systems and computer programs are typically expressed as functions or regions of the state space that evidence invariant properties of the system~\cite{hahn2019theory,floyd1993assigning}. 
Certificates for the quantitative and qualitative analysis of stochastic processes---known as supermartingale certificates---have been widely studied, 
especially in stochastic control with a focus on asymptotic stability, reachability, and avoidance objectives~\cite{kushner1967stochastic,meyn_tweedie_glynn_2009,douc2018markov}.
While traditionally these proof certificates are characterised analytically, hence requiring significant manual effort for their actual derivation, 
their automated construction has recently gained momentum due to advances in numerical methods~\cite{DBLP:conf/cdc/Papachristodoulou02,DBLP:conf/cdc/PrajnaJP04,DBLP:journals/tac/PrajnaJP07,parrilo2000structured},
as well as machine learning techniques for this purpose~\cite{DBLP:conf/concur/AbateEGPR23}. 
Automation in the construction of supermartingale certificates has stimulated their adoption in termination analysis~\cite{DBLP:conf/cav/ChakarovS13,DBLP:conf/popl/ChatterjeeNZ17,DBLP:journals/toplas/ChatterjeeFNH18,DBLP:conf/cav/ChatterjeeGMZ22,DBLP:conf/cav/AbateGR20,DBLP:conf/esop/MoosbruggerBKK21}, 
reachability, safety and reach-avoidance analysis~\cite{DBLP:conf/tacas/BatzCJKKM23,DBLP:journals/tecs/HuangCLYL17,DBLP:conf/atva/JagtapSZ18,DBLP:conf/pldi/WangS0CG21,DBLP:conf/amcc/LavaeiSF22},
 cost bound analysis~\cite{DBLP:conf/pldi/Wang0GCQS19,DBLP:journals/pacmpl/MoosbruggerSBK22,DBLP:conf/cav/SunFCG23,DBLP:journals/pacmpl/ChatterjeeGMZ24}, 
stochastic control synthesis and learning~\cite{DBLP:journals/tac/JagtapSZ21,DBLP:conf/aaai/LechnerZCH22,DBLP:conf/nips/ZikelicLVCH23,DBLP:journals/csysl/MathiesenCL23,DBLP:conf/aaai/ZikelicLHC23,DBLP:conf/cav/AbateGR24}. 

We present a general methodology for the formalisation of quantitative proof certificates for probabilistic systems and demonstrate 
its practical application in developing model checking and control synthesis algorithms. 
We show that every specification that falls within the class of {\em shift-invariant} events 
admits a stochastic invariant that bounds its probability from below. 
A stochastic invariant is a region of the state space associated with a supermartingale that is sufficient to 
bound from above the probability of leaving the invariant. We provide two converse theorems for their necessary existence:
for systems with general state space, we establish the existence of a stochastic invariant 
that is sufficiently strong to bound the probability of the shift-invariant specification up to arbitrary approximation; 
for systems with finite state space, we establish the existence of a stochastic invariant that quantifies its probability exactly.

Our result reduces the problem of computing a lower bound on the probability of a shift-invariant specification 
to the problem of computing a stochastic invariant alongside a 
proof certificate for the almost-sure satisfaction of the specification. Our reduction is complete up to arbitrary approximation for systems with general state space,  complete for systems with finite state space, and applies to a rich class of specifications. 
Shift-invariant specifications encompass Büchi and co-Büchi acceptance conditions, which have existing quantitative certificates~\cite{DBLP:journals/csysl/AjeleyeZ24a}, as well as Muller, parity, Rabin, and Streett conditions, for which no quantitative certificates have previously been presented. As such, our method not only unifies existing results but also lays the foundations for developing new quantitative supermartingale certificates. 

We instantiate our theory to the design and implementation of the first supermartingale certificate for the 
quantitative verification and control of $\omega$-regular and linear temporal logic (LTL) specifications. 
We leverage our theory alongside two existing results.
Firstly, $\omega$-regular and LTL specifications enjoy reduction to Streett acceptance conditions through 
composition with deterministic Streett automata~\cite{DBLP:conf/focs/Safra88}. 
Secondly, Streett acceptance conditions have supermartingale certificates for their 
almost-sure satisfaction with supporting invariants~\cite{DBLP:conf/cav/AbateGR24}.
Since Streett acceptance conditions are shift-invariant, our theory extends the existing 
supermartingale certificates for almost-sure Streett acceptance to additionally quantify lower and upper bounds on the acceptance probability. This enables the algorithmic $\omega$-regular quantitative verification and control of probabilistic systems with general state space, encompassing and generalising safety, reachability, reach-avoidance, recurrence, persistence properties and LTL.  

We demonstrate the practical efficacy of our method with a prototype 
for the simultaneous construction of parametrised supermartingale certificates alongside parametrised control policies expressed as polynomials of known degree. 
We leverage polynomial Positivstellensatz results to reduce it 
to a decision problem over the existential theory of the reals, amenable to satisfiability solving modulo quantifier-free non-linear real arithmetic~\cite{DBLP:journals/cca/JovanovicM12,DBLP:conf/cav/BjornerN24}. 
Our algorithm is sound and complete  
relatively to the existence of the almost-sure component of our certificates and up to a desired approximation error. 
We compute upper and lower probability bounds using polynomials of varying degree on several examples  
with infinite state space, which are beyond the reach of the existing tools. 

Our contribution is threefold. First, we present a general theory of quantitative supermartingale certificates,
which extends every certificate for almost-sure acceptance of shift-invariant specifications to their quantitative counterpart. 
Second, we introduce a special theory of quantitative Streett supermartingale certificates based on our methodology, 
which results in the first quantitative supermartingale certificate for $\omega$-regular specifications and LTL.
Third, we implement our theory in an algorithm for quantitative $\omega$-regular verification and control, 
and demonstrate its practical efficacy on examples.

\section{Stochastic Invariants}

We consider stochastic systems over general state space $(S, \Sigma)$, where $S$ denotes the set of states
and $\Sigma$ denotes the associated $\sigma$-algebra. We treat quantitative model checking and control synthesis 
problems for specifications over an infinite time horizon measured over $(\Omega, {\cal F})$, 
where the set of outcomes $\Omega = S^\omega$ are the infinite trajectories 
and the set of events ${\cal F} = \bigotimes_{i \in \omega} \Sigma_i$ (with $\Sigma_i = \Sigma$) 
are the measurable specifications. 
As is standard in stochastic analysis~\cite{meyn_tweedie_glynn_2009}, 
we rely on the result that 
every initial probability measure and transition probability kernel gives rise to a well-defined probability measure over specifications.

\begin{theorem}
\label{thm:meas-spec}
Let $\mu : \Sigma \to [0, 1]$ be an initial probability measure and 
$P : S \times \Sigma \to [0, 1]$ be a transition probability kernel. Then, there exists a stochastic process $\mathbf{\Phi} = ( \Phi_0, \Phi_1, \ldots )$ on the trajectory space $(\Omega, \mathcal{F})$ 
and a probability measure ${\sf P}_\mu : \mathcal{F} \to [0,1]$ where ${\sf P}_\mu(\mathbf{\Phi} \in L)$ is the probability that $\mc$ satisfies the specification $L \in {\cal F}$ and, for every $n \in \Nat$ and 
$A_0 \in \Sigma, \ldots, A_n \in \Sigma$, the following holds: 
\begin{multline}
        {\sf P}_\mu(\Phi_0 \in A_0 \land \dots \land \Phi_n \in A_n)
    =\\\int_{s_0 \in A_0} \cdots \int_{s_{n-1} \in A_{n-1}} 
    \mu(\mathrm{d}s_0)
    P(s_0, \mathrm{d}s_1) 
    \cdots 
    P(s_{n - 1}, A_n).\label{eqn:measure-traj}
\end{multline}
\end{theorem}

We frame our work around the operation of {\em time shift}, 
which encapsulates the forgetfulness of the process with respect to its past--i.e., the Markov property. 
We define the (time) shift operator $\theta$ as the measurable mapping on $\Omega$ 
\begin{equation}
    \theta(s_0, s_1, \ldots, s_n, \ldots)
    =
    ( s_1, s_2, \ldots, s_{n+1}, \ldots ).
\end{equation}
This characterises time-homogeneous Markov chains over general state spaces, 
our reference model throughout the paper unless stated otherwise. 
Also, henceforth we use $\delta_s : \Sigma \to [0, 1]$ to denote the Dirac measure at $s \in S$.

\begin{definition}[Time-Homogeneous Markov Chains] 
A time-homogeneous Markov chain is a stochastic process $\mc$ defined in terms of
an initial probability measure $\mu : \Sigma \to [0, 1]$ and 
probability transition kernel $P : S \times \Sigma \to [0, 1]$, 
having a natural filtration $\mathcal{F}_n^\Phi = \sigma\left( 
    \Phi_0, \ldots, \Phi_n
    \right) \subseteq \mathcal{F}$ satisfying the Markov property, i.e.,
\begin{equation}\label{eqn:time-hom-m-p}
    {\sf E}_\mu[H \circ \theta^n \mid \mathcal{F}_n^\Phi]
    =
    {\sf E}_{\delta_{\Phi_n}}[H]\qquad \text{a.s.}\ [{\sf P}_\mu]
\end{equation}
for every random variable $H$ on $(\Omega, {\cal F}, \Prob_\mu)$ and every $n \in \Nat$~\cite[p.~70]{meyn_tweedie_glynn_2009}.
\end{definition}

Time-homogeneity allows us to derive global properties of the stochastic process by 
locally reasoning about the transition probability kernel $P$ and the initial probability measure $\mu$. 
For this purpose, we define the post-expectation $(Ph) \colon S \to \Real$ and the init-expectation $(\mu h) \in \Real$ operations of any real-valued measurable function $h \colon S \to \Real$, with respect to the process, as follows:
\begin{equation}\label{def:post-init-expectation}
    Ph(s) = \int_{u \in S} 
    h(u) ~P(s, \mathrm{d}u), 
    \qquad \mu h = 
    \int_{s \in S} h(s)~\mu(\mathrm{d}s).
\end{equation}
These two operators are the essential elements in the formalisation and the construction of supermartingale certificates.
Specifically, the post-expectation $Ph(s)$ of the function $h$ at state $s \in S$ gives the expected value of $h$ at the next state conditional on $s$ being the current state; 
similarly, the init-expectation $\mu h$ gives the expected value of $h$ at the initial time.
The algorithmic synthesis of certificates relies on expressing the post- and init-expectation of value functions in a closed form, 
for which appropriate procedures are available~\cite{DBLP:conf/cav/GehrMV16}.

Our methodology leverages the proof rule for stochastic invariants, which is the most basic form of a quantitative supermartingale 
certificate~\cite[Theorem 1]{kushner1965stability}. 
A stochastic invariant is a region of the state space $I$ associated with a value function $V_0$ that bounds from above the probability that the process 
escapes $I$.
\begin{theorem}[Stochastic Invariants]\label{thm:invar}
    Suppose that there exists a measurable set $I \in \Sigma$ and a measurable function $V_0 \colon S \to \bbbr_{\geq 0}$ 
    such that
    \begin{align}
        \forall s \in I \colon &P V_0(s) \leq V_0(s),\label{eqn:invar-noninc}\\
        \forall s \notin I \colon &V_0(s) \geq 1.\label{eqn:invar-one}
    \end{align}
    Then, ${\sf P}_\mu({\mathbf \Phi} \notin I^\omega) \leq \mu V_0$. 
\end{theorem}

We show that stochastic invariants are sufficient to characterise the probability of the rich class of specifications 
that are invariant to time shift, i.e., the specifications that are invariant to addition or deletion of finite prefixes.
\begin{definition}[Shift-Invariant Specifications] 
A specification $L \in {\cal F}$ is invariant to time shift if it satisfies the following property:
\begin{equation}
    \theta^{-1} L = L.\label{eqn:shift-invariant}
\end{equation}
\end{definition} 

\begin{remark}[Connection to Tail Objectives]\label{rmk:tail}
    Specifications satisfying \cref{eqn:shift-invariant} are sometimes referred to as \emph{tail objectives}~\cite{DBLP:journals/tcs/Chatterjee07,DBLP:conf/icalp/KieferMSTW20,DBLP:conf/mfcs/ChatterjeeHH09}.
    In fact, every shift-invariant event is also a tail event, i.e., a member of the tail $\sigma$-algebra $\cap_{i\in \omega} ~\sigma(\Phi_i, \Phi_{i+1}, \dots)$. The converse is not true, and not every tail event is shift-invariant~\cite[p.~260]{douc2018markov}. \qed
\end{remark}
\begin{remark}[Connection to Liveness Properties~\cite{DBLP:journals/dc/AlpernS87}] \label{rmk:liveness}
    Shift invariance is strictly stronger than liveness.
    For example, consider the liveness property $L = \{ \exists n \in \Nat \colon \Phi_n \in A\}$, specifying that $A \in \Sigma$ 
    eventually happens. Under a shift we obtain $\theta^{-1}L = \{ \exists n \in \Nat \colon \Phi_{n+1} \in A\} \neq L$, 
    excluding the option to hit $A$ at time 0.
     \qed
\end{remark}

We address the question of determining the probability for which a time-homogeneous Markov chain $\mc$ satisfies a 
shift-invariant specification $L$ using supermartingale certificates. Our methodology is underpinned by 
the relation between a shift-invariant specification and the random variable characterising its satisfaction probability,
which we show in the following technical result.
\begin{theorem}\label{thm:language-levy}
Suppose that $L \in {\cal F}$ is shift-invariant.  Then
    \begin{equation}
       {\sf P}_\mu({\bf {\Phi}} \in L) = {\sf P}_\mu(~\inf_n {\sf P}_{\delta_{\Phi_n}}({\bf \Phi} \in L) > 0~).\label{eqn:language-levy}
    \end{equation} 
\end{theorem}
\begin{example}[Intuition for \cref{thm:language-levy}]\label{ex:language-levy}
Consider a Markov chain on the countable state space $S = \Nat$ as illustrated in \cref{fig:gambler}, defining a biased random walk that, at each time, increments the state with probability $0.51$, and otherwise decrements the state with probability $0.49$, unless it reaches the state 0, at which it remains thereafter.
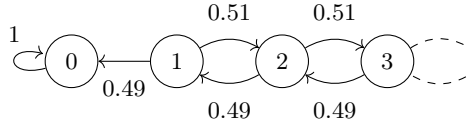
\begin{figure}[h]
\centering
    \begin{tikzpicture}[minimum size=7mm, node distance=14mm]
        \node[draw, circle] (0) {0};
        \node[draw, circle, right of=0] (1) {1};
        \node[draw, circle, right of=1] (2) {2};
        \node[draw, circle, right of=2] (3) {3};
        \node[right of=3] (4) {};
        
        \draw (0)  edge[->, loop left] node[above] {$1$} (0)
        (1) edge[->, bend left] node[above] {$0.51$} (2) 
        (2) edge[->, bend left] node[above] {$0.51$} (3)
        (3) edge[dashed, bend left] (4)
        (4) edge[dashed, bend left] (3);

        \draw 
        (3) edge[->, bend left] node[below] {$0.49$} (2)
        (2) edge[->, bend left] node[below] {$0.49$} (1)
        (1) edge[->] node[below] {$0.49$} (0);
    \end{tikzpicture}
    \caption{Gambler's Ruin.}\label{fig:gambler}
\end{figure}
Consider the event $L = \{ \sum_{n = 0}^{\infty} \ind_{ A }(\Phi_n) = \infty \}$,
which specifies that $A \in \Sigma$ is visited infinitely often (${\bf 1}_A$ denotes the indicator function of $A$).
Notably, this specification is shift-invariant because  
$\theta^{-1}L = \{ \sum_{n = 0}^{\infty} \ind_{A}(\Phi_{n+1}) = \infty \} = L$.
Suppose that $A = \{ 0 \}$. Then, for a state $x \in \Nat$, the probability that the Markov chain above satisfies $L$ corresponds to 
\begin{equation}
    \Prob_{\delta_x}(\mc \in L) = 
    \begin{cases}
        \left(49/51\right)^x &\text{if }x > 0,\\
        1 &\text{if }x = 0.
    \end{cases}
\end{equation}

Our main observation is that the expression $\Prob_{\delta_{\Phi_n}}(\mc \in L)$ defines a random variable on the probability space $(\Omega, {\cal F}, \Prob_\mu)$, which for shift-invariant properties is equal to the probability $\Prob_\mu(\mc \in L \mid {\cal F}_n^\Phi)$ of the system satisfying the specification, conditional on the information contained in the stochastic process up to time $n$. In fact, this random variable can be simulated in a computer program as we illustrate in \cref{fig:language-levy}, which shows 10 random simulations of the associated stochastic process under initial distribution $\mu = \delta_{10}$.
\begin{figure}[h]
    \begin{tikzpicture}
        \begin{axis}[xlabel={$n$}, ylabel={$\Prob_{\delta_{\Phi_n}}(\mc \in L)$},
            ymin=0, ymax=1,
            xmin=0, xmax=5001,
            xtick distance=1000,
            ytick distance=0.2,
            minor y tick num=1,
            scaled x ticks=false,
            xtick={0,1000,2000,3000,4000,5000},
            height=4cm, 
            width=\textwidth,
            ]
            \def\colorone{black!80}
            \def\colorzero{black!30}
            \def\colorbetween{black!55}
            \def\filename{examplecut.dat}

            \addplot [\colorone] table[x={n}, y={s1}] {\filename};
            \addplot [\colorone] table[x={n}, y={s2}] {\filename};
            \addplot [\colorone] table[x={n}, y={s3}] {\filename};
            \addplot [\colorzero] table[x={n}, y={s4}] {\filename};
            \addplot [\colorzero] table[x={n}, y={s5}] {\filename};
            \addplot [\colorone] table[x={n}, y={s6}] {\filename};
            \addplot [\colorzero] table[x={n}, y={s7}] {\filename};
            \addplot [\colorone] table[x={n}, y={s8}] {\filename};
            \addplot [\colorbetween] table[x={n}, y={s9}] {\filename};
            \addplot [\colorone] table[x={n}, y={s10}] {\filename};
        \end{axis}
    \end{tikzpicture}
    \caption{Random simulations of the satisfaction probability process of \cref{ex:language-levy}.}
    \label{fig:language-levy}
\end{figure}
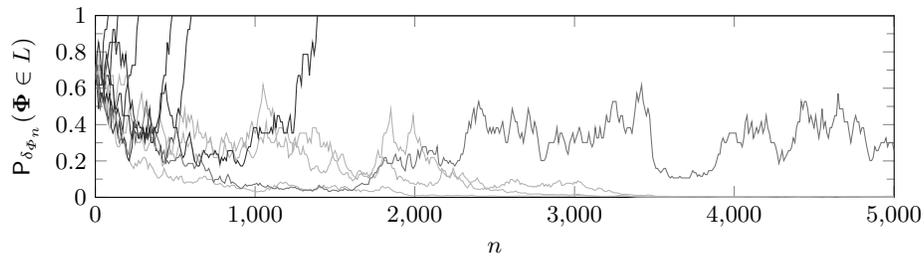

There are two central reasons why \cref{eqn:language-levy} holds. Firstly, the stochastic process $\Prob_{\delta_{\Phi_n}}(\mc \in L)$ is a non-negative martingale. This implies that, if the value of the stochastic process ever hits 0, it must remain at 0 at all times thereafter. 
Secondly, the process $\Prob_{\delta_{\Phi_n}}(\mc \in L)$ almost-surely converges to 1 when $\mc \in L$ and to 0 otherwise,
as a consequence of L{\'e}vy's 0-1 Law. 
Since $\Prob_{\delta_{\Phi_n}}(\mc \in L)$ converges to either 0 or 1 almost surely,
the probability of the variable $\Prob_{\delta_{\Phi_n}}(\mc \in L)$ converging to 1 corresponds to the probability of not converging to 0, and
since 0 is an absorbing value for a non-negative martingale,
this corresponds to the probability that its infimum is positive. \qed
\end{example}

A consequence of the relation between shift-invariant specifications and 
the random variables associated with their satisfaction probability is that, 
for every desired approximation error $\epsilon > 0$, 
we can always choose an appropriate level set of this random variable to define a sufficiently tight stochastic invariant.

\begin{theorem}\label{thm:approx-general}
    Suppose that $L \in {\cal F}$ is shift-invariant. Then, for every $\epsilon > 0$ there exists a measurable set $I \in \Sigma$ such that $\Prob_\mu({\bf \Phi} \in I^\omega \land {\bf \Phi} \notin L) = 0$ and
    \begin{equation}
        \Prob_\mu({\bf {\Phi}} \in L) - \epsilon \leq \Prob_\mu({\bf {\Phi}} \in I^\omega) \leq \Prob_\mu({\bf {\Phi}} \in L). 
        \label{eqn:approx-general}
    \end{equation}
\end{theorem}

\begin{example}[The Gambler's Ruin]\label{ex:gambler}
    The Markov chain in \cref{ex:language-levy} corresponds to the Gambler's Ruin problem~\cite[p.~345]{FellerVol1}. 
    It is a classic result that if the process starts from $x > 0$, the probability of hitting any other value $y > x$---equivalent to 
    the probability of exiting the set $I = \{0, \dots, y-1\}$---is given by
    \begin{equation}
        \Prob_{\delta_x}(\mc \notin I^\omega) = \frac{1 -  (49/51)^x}{1 - (49/51)^{y}}. 
    \end{equation}
    It follows that the probability of avoiding $y$ converges asymptotically, for increasing $y$, 
    to the probability of hitting 0:
    \begin{equation}
        \lim_{y \to \infty} \underbrace{1 - \frac{1-(49/51)^x}{1-(49/51)^{y}}}_{\Prob_{\delta_x}(\mc \in I^\omega)} = 
        \underbrace{\left( 49/51 \right)^x}_{\Prob_{\delta_x}(\mc \in L)}.
    \end{equation}
    This shows that, for every $\epsilon > 0$, there exists a sufficiently large $y$ such that $I$ satisfies \cref{eqn:approx-general} with $\mu = \delta_x$. 
    Moreover, for every $y$ the event of eventually hitting either 0 or $y$ has probability 1; in other words $\Prob_{\delta_x}({\bf \Phi} \in I^\omega \land {\bf \Phi} \notin L) = 0$. \qed
\end{example}

We further demonstrate that, for finite systems, a stochastic invariant that exactly quantifies the probability of the specification always exists.
\begin{theorem}\label{thm:exact-finite}
    Suppose that $L \in {\cal F}$ is shift-invariant and $S$ 
    is finite. Then, there exists a measurable set $I \in \Sigma$ such that $\Prob_\mu({\bf \Phi} \in I^\omega \land {\bf \Phi} \notin L) = 0$ and
    \begin{equation}
        \Prob_\mu({\bf {\Phi}} \in I^\omega) = \Prob_\mu({\bf {\Phi}} \in L). \label{eqn:exact-finite}
    \end{equation}
\end{theorem}

\begin{example}
    Assume that the Markov chain in \cref{ex:language-levy} has an upper bound $N > 0$ that is a sink state, making its state space finite. 
    Then, the event of hitting 0---which is the event $L$---corresponds exactly to the event of avoiding $N$---which is the event $I^\omega$ 
    with $I = \{0, \dots, N - 1\}$. Since the two events are equivalent, their probabilities are as well, satisfying \cref{eqn:exact-finite}. \qed
\end{example}

\begin{remark}[Existence of Value Functions]
    Our converse theorems establish the existence of invariant regions $I \in \Sigma$. 
    This implies the existence of appropriate value functions, which can be
    defined as $V_0(s) = \Prob_{\delta_s}(\mc \notin I^\omega)$. 
    These are necessarily measurable and are guaranteed to satisfy \cref{eqn:invar-noninc,eqn:invar-one}.\qed
\end{remark}

\section{Quantitative Supermartingale Certificates}

We propose a general methodology for the formalisation of proof rules to 
establish probability bounds for a broad variety of specifications. 
We show that the problem of computing lower bounds for the probability of satisfaction 
of shift-invariant specifications can be decomposed into two problems: 
computing a stochastic invariant alongside a lower bound for its probability, and 
deciding the almost-sure satisfaction of the specification conditional to the stochastic invariant. 

\begin{theorem}\label{thm:decomp}
    Suppose that $L \in {\cal F}$ is shift-invariant and, 
    for some probability bound $p \in (0, 1]$ and measurable set $I \in \Sigma$, the following two conditions hold:
    \begin{align}
        &{\sf P}_\mu(\mathbf{\Phi} \in I^\omega) \geq p, \label{eqn:pee} \\
        &{\sf P}_\mu (\mathbf{\Phi} \in L \mid \mathbf{\Phi} \in I^\omega) = 1. \label{eqn:one}
    \end{align}
    Then, ${\sf P}_\mu(\mathbf{\Phi} \in L) \geq p$.
\end{theorem}

This result enables the extension of every supermartingale certificate proof rule for almost-sure satisfaction, conditional 
to a deterministic invariant, towards a quantitative proof rule for the same specification. 
Specifically, our proof rule for stochastic invariants presented in \cref{thm:invar} provides the appropriate 
constraints for the formalisation of quantitative supermartingale certificates.
\begin{example}[{\protect A Proof Rule for Quantitative Termination \cite[Theorem 4]{DBLP:conf/cav/ChatterjeeGMZ22}}]\label{example:quant-reach-rule}
    Using our methodology, we formalise a supermartingale certificate proof rule for the quantitative finite-time 
    termination of probabilistic programs. For the set of terminal states $A \in \Sigma$, which are assumed to be sink states, this corresponds to determining the probability of $L = \{ \sum_{n=0}^\infty {\bf 1}_{A}(\Phi_n) = \infty \}$, which is shift-invariant.
    We combine the proof rule for ranking supermartingales~\cite[Definition 9]{DBLP:conf/cav/ChakarovS13} (cf. \cref{eqn:ranking})---which proves almost-sure termination in expected finite time---with \cref{thm:invar}, and obtain the following (known) proof rule:
    \begin{align}
        \forall s \in I \colon &PV_0(s) \leq V_0(s),\label{eqn:reach-drift}\\
        \forall s \notin I \colon &V_0(s) \geq 1,\label{eqn:reach-one}\\
        \forall s \in I \setminus A \colon & PV_1(s) \leq V_1(s) - \varepsilon. \label{eqn:ranking}
    \end{align}
    Here, any region $I \in \Sigma$, non-negative value functions $V_0,V_1\colon S\to\Real_{\geq 0}$, and 
    positive constant $\varepsilon > 0$ constitute a quantitative supermartingale certificate where $1 - \mu V_0$ is a lower bound 
    upon the probability of hitting target $A$.

    Consider the quantitative verification problem developed in \cref{ex:language-levy,ex:gambler}, 
    which corresponds to the termination question with terminal state $A = \{ 0 \}$. A valid supermartingale certificate is given by the following components:
    \begin{equation}
        V_0(x) = \frac{1 -  (49/51)^x}{1 - (49/51)^{y}},\quad V_1(x) = y-x,\quad I = \{0, \dots, y-1\},\quad \varepsilon = 0.02,\label{eq:term-cert}
    \end{equation}
    where $y \in \Nat$ is any value larger than the initial state. For initial state 10, 
    the true probability is approximately $0.6703 \approx (49/51)^{10}$. 
    With $y = 50$, we obtain bound $1 - V_0(10) \approx 0.62$; 
    with $y = 100$, we obtain the tighter bound $1 - V_0(10) \approx 0.66$; 
    with $y = 200$, we obtain the much tighter bound $1 - V_0(10) \approx 0.6702$. 
    Notably, the true probability $(49/51)^x$ would violate \cref{eqn:reach-one}, and in this example
    a bounded invariant is essential to construct a ranking supermartingale $V_1$. \qed
\end{example}

Our converse results presented in \cref{thm:approx-general,thm:exact-finite} guarantee that our methodology yields complete certificates
up to arbitrary approximation for systems with general state space, and complete certificates for finite systems.
\begin{theorem}[$\epsilon$-Completeness for General Markov Chains]\label{thm:epsilon-complete}
    Suppose that $L \in {\cal F}$ is shift-invariant. Then, for every arbitrary $\epsilon > 0$, 
    there exists a measurable set $I \in \Sigma$ such that 
    \cref{eqn:pee,eqn:one} hold with $p = {\sf P}_\mu(\mathbf{\Phi} \in L) - \epsilon$.
\end{theorem}
\begin{theorem}[Completeness for Finite Markov Chains]
\label{thm:complete}
    Suppose that $L \in {\cal F}$ is shift-invariant and $S$ is finite. 
    Then, there exists a measurable set $I \in \Sigma$ such that \cref{eqn:pee,eqn:one} hold with $p = {\sf P}_\mu(\mathbf{\Phi} \in L)$.
\end{theorem}
\begin{remark}\label{rmk:relative}
    Composing stochastic invariants and almost-sure certificates, as described in \cref{thm:decomp}, 
    results in complete proof rules for probabilistic lower bounds under the assumption that the proof rule for conditional almost-sure satisfaction is complete. In other words, all completeness guarantees of the proof rule for almost-sure satisfaction carry over to their quantitative extension, up to approximation or exactly, as described in \cref{thm:epsilon-complete,thm:complete} respectively. \qed
\end{remark}

Our methodology generalises and unifies existing proof rules for quantitative model checking and control synthesis, 
while providing the foundation for formalising quantitative supermartingale certificates for new specifications and objectives. 
It applies to the rich class of shift-invariant specifications, which includes and extends beyond a broad variety of special cases.
This includes specifications defined as limits\cite[Lemma 5.1.6]{douc2018markov}, such as the limit objectives on cost and reward considered in reinforcement learning, and asymptotic stability considered in control theory, all of which are also tail events (cf.~\cref{rmk:tail}). 
Moreover, it also includes B\"uchi, co-B\"uchi, Rabin, Streett, Muller, and parity acceptance conditions of automata over infinite words~\cite{DBLP:journals/tcs/Chatterjee07}. 
As we demonstrate, this enables in particular the development of quantitative supermartingale certificates for $\omega$-regular specifications.

\section{Quantitative $\omega$-Regular Verification and Control}

We present the first quantitative supermartingale certificate for $\omega$-regular specifications, 
which we obtain as a result of \cref{thm:invar,thm:decomp} 
and the supermartingale certificate for the almost-sure acceptance of Streett conditions~\cite{DBLP:conf/cav/AbateGR24}. 

An $\omega$-regular specification (or language) over a finite set of atomic propositions $\Pi$, 
which we define as predicates over the state space of the system under analysis, corresponds to the language
of an $\omega$-regular expression
whose alphabet is the Boolean truth valuations of $\Pi$.
An important class of $\omega$-regular specifications is the temporal behaviour described using linear temporal logic (LTL).
An LTL formula $\varphi$ extends propositional logic (over the atomic propositions $\Pi$) with the temporal  
{\em next} operator ${\sf X}\varphi$, indicating that $\varphi$ holds after one step in the future, 
the {\em eventually} operator ${\sf F}\varphi$, indicating that $\varphi$ holds at some point in the future, 
the {\em always} operator ${\sf G} \varphi$, indicating that $\varphi$ holds at all times in the future, 
and the {\em until} operator $\varphi {\sf U} \psi$, indicating that $\varphi$ holds at all times in the future before $\psi$, which in turn holds at some point in the future~\cite{DBLP:conf/focs/Pnueli77}. 

We treat the problem of determining the probability of satisfying an $\omega$-regular specification over $\Pi$ 
for a system under analysis 
whose semantics is a time-homogeneous Markov chain $\hat{\mc}$ 
with general state space 
$(\hat{S}, \hat{\Sigma})$, initial probability measure $\hat{\mu}$ and transition probability kernel $\hat{P}$.
The problem is defined in terms of 
a measurable labelling function $\llangle \cdot \rrangle \colon \hat{S} \to {\cal P}(\Pi)$ where
$\llangle s \rrangle \subseteq \Pi$ indicates the set of atomic propositions  
that hold true in state $s \in \hat{S}$, which we call the labelling of $s$, 
and interpret the $\omega$-regular specification according to its usual semantics 
over the set of traces ${\cal P}(\Pi)^\omega$. 
Notably, $\omega$-regular specifications lack shift invariance.
For example, the LTL formula $\varphi = {\sf F}a$, defining the event
${L}_{\varphi} = \{ \exists n \in \Nat \colon a \in \llangle \Phi_n \rrangle\}$, is not invariant to time shift (cf. \cref{rmk:liveness}).

Automata over infinite words reduce $\omega$-regular specifications to equivalent acceptance conditions that are shift-invariant, by extending the state space with additional memory which is given by the states of an $\omega$-automaton.
Büchi automata are the canonical example, but they require the presence of 
non-determinism, with which standard probability theory is limited.
Conversely, automata with Muller, Rabin, parity, 
and Streett acceptance conditions recognise   $\omega$-regular languages in their deterministic form~\cite{DBLP:conf/dagstuhl/2001automata}, which preserves the probabilistic nature of the system.
We consider the case of Streett automata, 
and generalise the existing supermartingale certificates for their almost-sure acceptance (from the literature~\cite{DBLP:conf/cav/AbateGR24}) 
to additionally produce lower and upper probability bounds for $\omega$-regular specifications. 
\begin{definition}[Deterministic Streett Automata]
A deterministic Streett automaton (DSA) over the finite set of propositions $\Pi$ consists of 
a finite set of states $Q$,  
an initial state $q_0 \in Q$, 
a transition function $T \colon Q \times {\cal P}(\Pi) \to Q $, 
and an acceptance condition $(F_1, G_1), \dots, (F_k,G_k)$ where $F_i, G_i \subseteq Q$ for $i = 1, \dots k$.
An infinite input trace $(p_0, p_1, p_2, \dots ) \in {\cal P}(\Pi)^\omega$ is accepted if there exists 
an infinite run $(q_0, q_1, q_2, \dots) \in Q^\omega$ such that $q_{n+1} = T(q_n, p_n)$ for every $n \in \Nat$ and,  
for every $i = 1, \dots, k$, either $\sum_{n=0}^\infty {\bf 1}_{F_i}(q_n) < \infty$ or $\sum_{n=0}^\infty {\bf 1}_{G_i}(q_n) = \infty$.
\end{definition}

There are multiple algorithms for the automatic construction of DSA,
in particular from LTL formulae~\cite{DBLP:conf/cav/KretinskyMSZ18,DBLP:conf/cav/Duret-LutzRCRAS22}. 
Given a DSA, the original $\omega$-regular verification 
question reduces to a question of Streett acceptance over the synchronous composition between the system under analysis $\hat{\mc}$ and the automaton. The synchronous composition is a Markov chain over state space $S = \hat{S} \times Q$ 
with the $\sigma$-algebra $\Sigma = \hat{\Sigma} \otimes {\cal P}(Q)$,  whose transition probability kernel $P \colon S \times \Sigma \to [0,1]$ and initial probability measure $\mu \colon \Sigma \to [0,1]$ 
are defined as follows:
\begin{align}
    P
    ((s, q), A)
    &= \int_{(u,r) \in A}
    \hat{P}(s, {\rm d}u)
    \cdot 
    \mathbf{1}_{\{r\}}(T(q, \llangle s \rrangle)), \\
    \mu(A) &= \int_{(u,r) \in A} \hat{\mu}({\rm d}u) \cdot {\bf 1}_{\{ r \}}(q_0).
\end{align}
This is associated with the Streett acceptance condition $(A_1,B_1) \in \Sigma^2, \dots, (A_k,\allowbreak B_k) \in \Sigma^2$ 
defined as $A_i = \hat{S} \times F_i, B_i = \hat{S} \times G_i$ for $i = 1, \dots k$.
\begin{remark}[Streett Acceptance is Shift-Invariant]
    As we establish in \cref{ex:language-levy}, the B\"uchi acceptance condition $\{ \sum_{n = 0}^{\infty} \ind_{ A }(\Phi_n) = \infty \}$ 
    is shift-invariant. We note that shift-invariant events are closed under countable Boolean operations~\cite[Proposition 5.1.5]{douc2018markov}, 
    and that Streett acceptance corresponds to the event $\cap_{i=1}^k (\{ \sum_{n = 0}^{\infty} \ind_{ A_i }(\Phi_n) = \infty \}^{\sf c} \cup \{ \sum_{n = 0}^{\infty} \ind_{ B_i }(\Phi_n) = \infty \})$.\qed
\end{remark}
\begin{figure}[h]
    \centering
    \begin{tikzpicture}[minimum size=7mm, node distance=14mm]
        \node[draw, circle] (0) {0};
        \node[draw, circle, right of=0] (1) {1};
        \node[draw, circle, right of=1] (2) {2};
        \node[right of=2] (3) {};
        
        \draw (0)  edge[->, bend left] node[above] {$0.51$} (1)
        (1) edge[->, bend left] node[above] {$0.51$} (2) 
        (2) edge[dashed, bend left] (3)
        (3) edge[dashed, bend left] (2)
        (2) edge[->, bend left] node[below] {$0.49$} (1)
        (1) edge[->, bend left] node[below] {$0.49$} (0)
        (0) edge[->, loop left] node[left] {0.49} (0);

        \node[state, right=14mm of 3] (q0) {$q_0$};
        \node[state, right=1cm of q0] (q1) {$q_1$};
        \draw (q0) edge[->, loop above] node[above] {$x \neq 0$} 
        (q0) edge[->] node[above] {$x = 0$} (q1)
        (q1) edge[->, loop above] node[above] {true} (q1)
        ($(q0)-(10mm,0)$) edge[->] (q0);

        \node at ($(q0)!0.5!(q1)-(0,1cm)$) {$F_1 = Q, G_1 =\{q_1\}$};
    \end{tikzpicture}
    \caption{A biased random walk over $\Nat$ and a DSA for the LTL specification ${\sf F}(x = 0)$.}
    \label{fig:reach-to-term}
\end{figure}
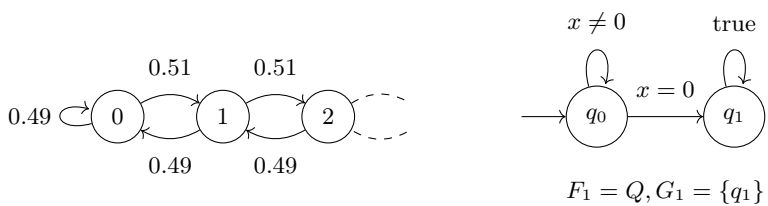
\begin{example}[Reachability as Recurrence]\label{ex:reachAsTerm}
Consider the biased random walk over state space $\hat{S} = \Nat$ 
illustrated in \cref{fig:reach-to-term}, which on all states increments with probability 0.51 and 
decrements with probability 0.49. 
Consider the LTL formula $\varphi = {\sf F}(x = 0)$, 
requiring that the process hits $0$ at least once. This is not shift-invariant (cf. \cref{rmk:liveness}). 
A DSA for this $\omega$-regular specification has two states $Q = \{ q_0, q_1 \}$ as depicted in \cref{fig:reach-to-term}, where $q_1$ is a sink state that is entered exactly when the random walk hits value 0. The acceptance condition of this automaton requires that $q_1$ is visited infinitely often---recurrence---which is shift invariant. 
Notably, their synchronous composition results in a Markov chain where every state $\{ q_1 \} \times \Nat$ essentially indicates that value $0$ 
has been hit at least once in the past. 
As a result, visiting $q_1$ infinitely often is equivalent to visiting 0 at least once,   
and has reduced our reachability question to an equivalent recurrence question. 
This is analogous to the termination problem developed in \cref{example:quant-reach-rule}, 
which in fact requires the terminal state 0 to be a recurrent sink state. \qed 
\end{example}

Our new (and the first) quantitative supermartingale certificate for $\omega$-regular specifications combines   
the proof rule for stochastic invariants in \cref{thm:invar}
with the following (known) proof rule for the almost-sure acceptance of Streett conditions over general state space.
\begin{theorem}[{\protect Streett Supermartingales~\cite{DBLP:conf/cav/AbateGR24}}]\label{thm:streett}
    Let $ (A_1, B_1) \in \Sigma^2, \dots,\allowbreak (A_k, \allowbreak B_k)\in \Sigma^2$ be a Streett acceptance condition. Suppose that $ \Prob_\mu(\mc \in I^\omega) > 0$ and 
    there exist $k$ measurable functions $V_1, \dots, V_k \colon S \to \Real_{\geq 0}$ such that
    \begin{align}
        &\forall s \in I \cap (A_i \setminus B_i) \colon 
    PV_i(s) \leq V_i(s) - \varepsilon,\label{eq:dec}\\
&\forall s \in I \cap B_i \colon
    PV_i(s) \leq V_i(s) + M,\label{eq:inc}\\
&\forall s \in I \setminus (A_i \cup B_i) \colon 
    PV_i(s) \leq V_i(s),\label{eq:noninc}
    \end{align}
    for some constants $\varepsilon, M > 0$. Then,
    \begin{equation}\label{eqn:conditional-1}
        \Prob_\mu \left( 
        \bigwedge_{i = 1}^k \sum_{n = 0}^{\infty} \ind_{A_i}(\Phi_n) < \infty \lor 
        \sum_{n = 0}^{\infty}\ind_{B_i}(\Phi_n) = \infty
        \mid \mc \in I^\omega \right) = 1. 
    \end{equation}
\end{theorem}

Our quantitative $\omega$-regular supermartingale certificate
requires synchronous composition between the system and a DSA recognising the same specification. 
Suppose that $k$ is the number of pairs in the acceptance condition.
Then, we require the simultaneous construction of a measurable set $I \in \Sigma$ and 
a sequence of measurable functions $V_0, \dots, V_k \colon S \to \Real_{\geq 0}$ 
such that $V_0$ satisfies \cref{eqn:invar-one,eqn:invar-noninc} and  
$V_i$ satisfies \cref{eq:dec,eq:inc,eq:noninc} for every $i = 1, \dots, k$. 
As a consequence of \cref{thm:invar,thm:streett,thm:decomp}, we have that $1 - \mu V_0$ is a lower bound 
on the probability that the system under analysis satisfies the $\omega$-regular specification.
\begin{theorem}[Quantitative Streett Supermartingales]\label{thm:quantitative-streett}
    Let $(A_1, B_1) \in \Sigma^2, \dots,\allowbreak (A_k, \allowbreak B_k)\in \Sigma^2$ be a Streett acceptance condition. 
    Suppose that there exists a measurable set $I \in \Sigma$ and $k+1$ measurable functions $V_0, \dots, V_k \colon S \to \Real_{\geq 0}$
    such that the following conditions hold:
    \begin{align}
        &\forall s \in I \colon P V_0 (s) \leq V_0(s),\label{qs:SI-dec}\\
        &\forall s \notin I \colon V_0(s) \geq 1,\label{qs:SI-ind}\\
        &\forall s \in I \cap (A_i \setminus B_i) \colon 
    PV_i(s) \leq V_i(s) - \varepsilon &\text{for }i = 1, \ldots, k,\label{qs:dec}\\
&\forall s \in I \cap B_i \colon
    PV_i(s) \leq V_i(s) + M &\text{for }i = 1, \ldots, k,\label{qs:inc}\\
&\forall s \in I \setminus (A_i \cup B_i) \colon 
    PV_i(s) \leq V_i(s) &\text{for }i = 1, \ldots, k,\label{qs:noninc}
    \end{align}
    for some constants $\varepsilon, M > 0$. Then,
    \begin{equation}
        \Prob_\mu \left( 
        \bigwedge_{i = 1}^k \sum_{n = 0}^{\infty} \ind_{A_i}(\Phi_n) < \infty \lor 
        \sum_{n = 0}^{\infty}\ind_{B_i}(\Phi_n) = \infty
        \right) \geq 1 - \mu V_0.
    \end{equation}
\end{theorem}
Our methodology similarly applies to alternative acceptance conditions (see \cref{rmk:relative}), 
such as Rabin, parity and Muller automata, but requires a proof rule for almost sure acceptance of these conditions.

\begin{figure}[h]
    \centering
    \begin{tikzpicture}[node distance=3cm]
        \node[state] (0) {$q_0$};
        \node[state, left of=0] (1) {$q_1$};
        \node[state, right of=0] (2) {$q_2$};
        \draw ($(0)-(7mm,7mm)$) edge[->] (0);
        \draw (0) edge[->] node[above] {$x < 10$} (1);
        \draw (0) edge[->,bend left] node[above] {$x \geq 100$} (2);
        \draw (2) edge[->,bend left] node[below] {$10 \leq x < 100$} (0);
        \draw (2) edge[->, loop above] node[above] {$x \geq 100$} (2);
        \draw (0) edge[->, loop above] node[above] {$10 \leq x < 100$} (0);
        \draw (1) edge[->, loop above] node[above] { true } (1);

        \draw (2) edge[->, bend left=60] node[above] {$x < 10$} (1);
        
        \node[right=5mm of 2] {$F_1 = \{q_0,q_1\}, G_1 =\emptyset$};
        \node[below=1mm of 0] {\color{blue}$\varepsilon$-dec};
        \node[below=1mm of 1] {\color{blue}$\varepsilon$-dec};
        \node[below=1mm of 2] {\color{gray}0-inc};
    \end{tikzpicture}
    \caption{A DSA for the LTL specification $(x \geq 10) {\sf UG} (x \geq 100)$.}
    \label{fig:brwgtt}
    \label{fig:enter-label}
\end{figure}
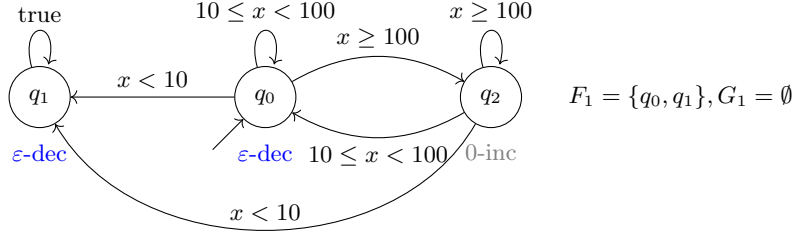
\begin{example}[Becoming Rich Without Getting Too Thin]\label{ex:brwgtt}
Consider the Gambler's Ruin model of \cref{fig:gambler}, or similarly the random walk of \cref{fig:reach-to-term}.
Consider the specification for which the amount $x$ eventually persists above 100 without ever going below 10.
This is a stabilise-while-avoid requirement 
specified as the LTL formula $\varphi = (x \geq 10) {\sf UG} (x \geq 100)$, 
and corresponds to the language accepted by the DSA in \cref{fig:brwgtt}. 
Our proof rule requires a region $I$ and two value functions $V_0$ and $V_1$ 
that simultaneously satisfy \cref{eqn:invar-noninc,eqn:invar-one,eq:dec,eq:noninc} on 
the synchronous composition. Firstly, we observe that it is impossible for any non-negative function $V_1$
to indefinitely decrease in the sink state $q_1$. Therefore, we must present a region $I$ that excludes $q_1$ and 
characterises every other reachable state, which we associate with a function $V_0$ that bounds from above the probability of leaving $I$:
\begin{align}
    V_0(x,q) &= 
    \begin{cases}
    \left( \frac{49}{51} \right)^{x-9} 
    &\text{if }(x, q) \in I\\
    1 &\text{otherwise},
    \end{cases} ~
    I 
    = 
    \left\{ (x,q)\colon
    \begin{array}{l}
    (q = q_0 \land 9 \leq  x \leq 100 \\
    \lor~(q = q_2 \land 99 \leq x)
    \end{array}
    \right\}.
\end{align}
Secondly, we observe that the expected value of $V_1$ must decrease by $\varepsilon$ in $q_0$ while never increasing in $q_2$.
We present a function with negative drift $PV_1(x,\cdot) - V_1(x,\cdot) < 0$ and choose an $\varepsilon > 0$ that upper-bounds the drift on $q_0$, which essentially indicates almost-sure finite permanence within $q_0$ conditional to $I^\omega$: 
\begin{align}
V_1(x, q) = 1 - e^{-x},  \quad \varepsilon = PV_1(100,\cdot) - V_1(100,\cdot).
\end{align}
As a result, we obtain the lower bound $1 - V_0(x,q_0) \leq \Prob_{\delta_x}(\mc \in L_{\varphi})$ on the 
probability of satisfying $\varphi$ from any initial state $x \in \Nat$. \qed
\end{example}

\begin{remark}[Upper Probability Bounds]
    Our quantitative $\omega$-regular supermartingale certificates also produce upper probability bounds. As $\omega$-regular languages are closed under complementation, 
    it suffices to compute a lower bound for the complementary specification. 
    According to the original representation of the specification, this requires the use of an appropriate complementation procedure~\cite{DBLP:conf/focs/Safra88}. 
    In the special case of LTL, it is sufficient to negate the formula. 
    \qed
\end{remark}

\begin{example}
    Consider the LTL verification problem of \cref{ex:brwgtt}. 
    A DSA for the complementary property $\lnot \varphi$ has the same structure of the automaton in \cref{fig:brwgtt}, 
    but has the alternative acceptance condition $F_1 = Q$ and $G_1 = \{ q_0, q_1 \}$. 
    This requires presenting a region $\bar{I}$ and value function $\bar{V}_0$ satisfying \cref{eqn:invar-noninc,eqn:invar-one} 
    and, as a consequence of \cref{eq:dec,eq:inc}, a value function $\bar{V}_1$ whose expected value must decrease by at least $\bar{\varepsilon} > 0$ on $q_2$ 
    and can increase by at most $M > 0$ on $q_0$ and $q_1$. Given any chosen bound $y \geq 9$ on the invariant, we present 
    \begin{equation}
        \bar{V}_0(x,q) = \begin{cases}
            \frac{1 -  (49/51)^{x-9}}{1 - (49/51)^{y-9}} & \text{if }x \geq 10 \land q \neq q_1\\
            0 &\text{otherwise},
        \end{cases}~ \bar{I} = \{ (x,q) \colon 0 \leq x \leq  y - 1 \},
    \end{equation}
    $\bar{V}_1(x,q) = y - x$, and $\bar{\varepsilon} = 0.02$. As a result, for every $x \in \Nat$ 
    we have the probability upper bound $\Prob_{\delta_x}(\mc \in L_\varphi) \leq \bar{V}_0(x,q_0)$. 
    Similarly to \cref{eq:term-cert},
    the tightness of $\bar{V}_0$ improves as $y$ increases. 
    For example, suppose the initial state is $50$. Under conservative numerical approximation, 
    we obtain $0.80 \leq \Prob_{\delta_{50}}(\mc \in L_\varphi) \leq 0.83$ with $y = 100$ and 
    $0.80 \leq \Prob_{\delta_{50}}(\mc \in L_\varphi) \leq 0.81$ with $y = 200$. \qed
\end{example}

Our proof rule reduces the quantitative $\omega$-regular model checking question to the 
problem of computing an appropriate region $I$ and appropriate value functions $V_0, \dots V_k$
satisfying the 
conditions of \cref{eqn:invar-noninc,eqn:invar-one,eq:dec,eq:inc,eq:noninc}. 
This extends to quantitative synthesis of parametrised control policies for stochastic processes whose 
transition probability kernel is conditional on control inputs, i.e., Markov decision processes (MDPs).
This is the problem of finding the parameters of a parametrised control policy for which the system satisfies an $\omega$-regular objective with a sufficiently high probability, 
which we reduce to the simultaneous synthesis of control parameters together with appropriate certificates.

\section{Algorithmic Synthesis of Supermartingale Certificates}\label{sec:algo}

The construction of certificates is a central objective in verification and control, 
supported by numerous algorithms for the automated synthesis of invariants and Lyapunov functions.
One standard approach for this purpose is to restrict the search within a specific class of parametrised function templates, 
which reduces their synthesis to the problem of computing appropriate parameters. 

We consider the problem of computing appropriate parameters $\zeta_0 \in Z_0, \dots, \allowbreak\zeta_k \in Z_k$ 
for the parametrised value functions $V_i \colon Z_i \times S \to \Real_{\geq 0}$ for $i = 0, \dots, k$, 
parameter $\eta \in H$ for the parametrised constraint $I \colon H \times S \to \{ \text{true}, \text{false} \}$, 
as well as control parameter $\kappa \in K$ for the parametrised transition probability kernel $P \colon K \times S \times \Sigma \to [0,1]$. 
In other words, we introduce functional templates for our stochastic invariant and Streett supermartingales, 
and assume a parametrised controller that governs the system behaviour according to its control parameter; 
the control parameter is constant throughout the system execution, 
whereas the control input varies over time as determined by the control policy. 

This results in a parametric model checking problem that encompasses the
quantitative $\omega$-regular control synthesis of  
memory-less parametrised control policies $\pi \colon K \times S \to U$
over MDPs with general state space $S$ and general control inputs $U$.
Given an MDP with kernel $\hat{P} \colon S \times U \times \Sigma \to [0,1]$ conditional on input,  
it suffices to express our parametrised transition kernel as $P(\kappa; s, A) = \hat{P}(s, \pi(\kappa; s), A)$. 
This also includes finite-memory parametrised policies, where it is required to augment $S$ with sufficient memory.  
The quantitative model checking of closed systems is the special case where $|K| = 1$. 

Given a desired lower probability bound $p \in (0,1]$, our objective is to compute values for the 
parameters $\zeta_0 \in Z_0, \dots, \allowbreak\zeta_k \in Z_k$, $\eta \in H$, $\kappa \in K$ and $\varepsilon,M > 0$ 
such that $p \leq 1- \mu V_0(\zeta_0)$ and the following universally quantified first-order logic formulae hold true: 
\begin{align}
    \forall s \in S &\colon I(\eta; s) \implies PV_0(\zeta_0, \kappa; s) \leq V_0(\zeta_0; s),\label{eqn:constraint1}\\
    \forall s \in S &\colon \lnot I(\eta; s) \implies V_0(\zeta_0; s) \geq 1,\label{eqn:constraint2}\\
   \forall s \in (A_i \setminus B_i) &\colon I(\eta; s) \implies PV_i(\zeta_i, \kappa; s) \leq V_i(\zeta_i; s) - \varepsilon & \text{for }i = 1, \dots, k,\label{eqn:constraint3}\\
    \forall s \in B_i &\colon I(\eta; s) \Rightarrow PV_i(\zeta_i, \kappa; s) \leq V_i(\zeta_i; s) + M & \text{for }i = 1, \dots, k,\label{eqn:constraint4}\\
    \forall s \notin A_i \cup B_i &\colon I(\eta; s) \implies PV_i(\zeta_i, \kappa; s) \leq V_i(\zeta_i; s) & \text{for }i = 1, \dots, k.\label{eqn:constraint5}
\end{align}
We require that post-expectations $(PV_i) \colon Z_i \times K \times S \to \Real_{\geq 0}$ and 
init-expectation $(\mu V_0)\colon Z_0 \to \Real_{\leq 0}$ are expressed in closed form, for which appropriate procedures exist~\cite{DBLP:conf/cav/GehrMV16}. 
Then, any algorithm that finds a satisfying assignment for the free parameters $\zeta_0, \dots, \zeta_k , \eta, \kappa, \epsilon$ and $M$ 
for the first-order formulae \cref{eqn:constraint1,eqn:constraint2,eqn:constraint3,eqn:constraint4,eqn:constraint5} would suffice.
According to the resulting form of the formulae above, one may select an appropriate decision procedure for this purpose. 

There are multiple approaches to solving the parameter synthesis problem expressed above. 
Firstly, we observe that the problem is decidable when the value functions and their 
expected values are expressed as polynomials of known degree and the constraints are 
expressed as semi-algebraic sets~\cite{tarski1998decision}. As a consequence, we have a
relatively complete algorithm under these assumptions,
in the sense that if polynomial certificates with sufficient precision on the probability bound exist and their degree is known 
then we have an algorithm to compute their coefficients. Under the additional assumption that $S$ is compact,
then polynomials with sufficiently high degree necessarily exist and we obtain complete algorithms for (relative to the existence of the almost-sure component $V_1, \dots, V_k$, see~\cref{rmk:relative})
that refine lower and upper bounds incrementally until a desired approximation gap is attained, 
leveraging the guarantees of \cref{thm:epsilon-complete,thm:complete}.

Decision procedures for quantified polynomial formulae are computationally intensive and, 
while theoretically feasible, pursuing arbitrary bounds is often impractical. 
A more practical approach (not reliant on compactness) is to select a polynomial template of desired degree 
while minimising the gap between upper and lower bounds on the probability of satisfaction, and 
possibly increase the degree until an allocated time budget is exhausted. 
Although this practical strategy is incomplete in general, in the sense that it may stop with trivial bounds, 
it is sound and produces useful results with sufficient time budget.

\begin{figure}[h]
    \centering
    \begin{tikzpicture}[baseline={(current bounding box.center)}]
        \begin{axis}[xlabel={$x$},
            ymin=0, ymax=1,
            xmin=10, xmax=100,
            xtick distance=1000,
            ytick distance=0.2,
            minor y tick num=1,
            scaled x ticks=false,
            xtick={10,20,30,40,50,60,70,80,90,100},
            height=5cm, 
            width=\textwidth*0.5,
            legend pos=south east,
            legend image post style={xscale=0.5}
            ]
            
            \addplot [black ] table[x={x}, y={s1}] {exampleBecomingRichOnce.dat};
            \addlegendentry{${\sf P}_{\delta_x}(L_\varphi)$}

            \addplot [name path=curve1, magenta, forget plot] table[x={x}, y={s5}] {upperboundBecRichOnce.dat};
            
            \addplot [name path=curve2, magenta] table[x={x}, y={s5}] {exampleBecomingRichOnce.dat};
            \addlegendentry{$d=5$}
            \addplot[magenta!20, forget plot] fill between[of=curve1 and curve2];
            
            \addplot [ cyan ] table[x={x}, y={s4}] {exampleBecomingRichOnce.dat};
            \addlegendentry{$d=4$}
            \addplot [ blue ] table[x={x}, y={s3}] {exampleBecomingRichOnce.dat};
            \addlegendentry{$d=3$}
            \addplot [ red ] table[x={x}, y={s2}] {exampleBecomingRichOnce.dat};
            \addlegendentry{$d=2$}

        \end{axis}  
    \end{tikzpicture}
    \quad
    \begin{tikzpicture}[node distance=2.1cm, baseline={(current bounding box.center)}]
        \node[state] (0) {$q_0$};
        \node[state, left of=0] (1) {$q_1$};
        \node[state, right of=0] (2) {$q_2$};
        \draw ($(0)-(7mm,7mm)$) edge[->] (0);
        \draw (0) edge[->] node[above] {$x < 10$} (1);
        \draw (0) edge[->] node[above] {$x \geq 100$} (2);
        \draw (2) edge[->, loop above] node[above] {true} (2);
        \draw (0) edge[->, loop above] node[above] {$10 \leq x < 100$} (0);
        \draw (1) edge[->, loop above] node[above] { true } (1);
        
        \node[below=1mm of 0,inner sep=0] {{\color{blue}$\varepsilon$-dec}};
        \node[below=1mm of 1,inner sep=0, xshift=-2mm] {$\varphi\equiv$~{\color{blue}$\varepsilon$-dec}};
        \node[below=1mm of 2,inner sep=0] {\color{red}M-inc};
        \node[below=5mm of 0,inner sep=0] {\color{red}M-inc};
        \node[below=5mm of 1,inner sep=0, xshift=-2.5mm] {$\lnot\varphi\equiv$~\color{red}M-inc};
        \node[below=5mm of 2,inner sep=0] {\color{blue}$\varepsilon$-dec};
        
    \end{tikzpicture}
    \caption{Polynomial approximations (and DSA) for the probability that the Gambler's Ruin in \cref{fig:gambler}
    satisfies $\varphi = (x \geq 10){\sf U}\allowbreak(x \geq 100)$; $d$ indicates the degree of the polynomial. 
    }
    \label{fig:richonce}
\end{figure}
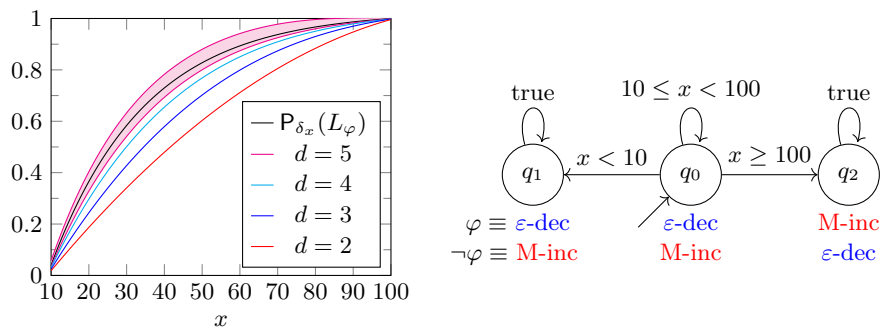
\begin{example}[Becoming Rich Once]
Consider the Gambler's Ruin example of \cref{fig:gambler} and the specification 
for which the process exceeds 100 before possibly falling below 10,
i.e., $\varphi = (x \geq 10){\sf U}\allowbreak(x \geq 100)$. 
This is recognised by the DSA in \cref{fig:richonce}, with the acceptance condition $F_1 = Q, G_1 = \{q_2\}$. 
This acceptance condition requires avoidance of $q_1$ and finite permanence in $q_0$, while imposing no restrictions on $q_2$. 
We assume $\mu = \delta_{50}$ and optimise the bounds accordingly. 
For the invariant region, we associate each automaton state with a parametrised semi-algebraic set, 
and for this example we obtain the rectangular region associating $q_0$ with the interval $[9,100]$, $q_1$ with the empty set, and $q_2$ with $[0, \infty)$. 
For each value function $V_i$, we adopt a polynomial template of degree $d$, whose coefficients are piecewise-defined according 
to the automaton state $q_j$:
\begin{equation}
    V_i(c_{i,j,0}, \dots, c_{i,j,d}; x,q_j) = c_{i,j,0} + c_{i,j,1}\cdot x + c_{i,j,2}\cdot x^2 + \dots + c_{i,j,d}\cdot x^d. 
\end{equation}
We obtain a piecewise-defined linear Streett supermartingale function given by $V_1(x,q_0) = 101 - x$, $V_1(x,q_1) = 101$, $V_1(x,q_2) = 0$, 
along with a piecewise-defined polynomial value function $V_0(x, q_1) = 1$, $V_0(x,q_2) = 0$, and 
 higher-degree polynomials for $V_0(x, q_0)$, yielding the lower probability bounds depicted in \cref{fig:richonce} for the degrees $d = 2,3,4, 5$. 
As shown, the lower probability bound becomes increasingly tighter with higher polynomial degrees. 

We further compute a polynomial upper approximation on the probability of satisfying $\varphi$ 
by computing a dual lower approximation on the probability of satisfying $\lnot \varphi$. 
This specification corresponds to the DSA of \cref{fig:richonce} with the acceptance 
condition $F_1 = Q, G_1 = \{q_0,q_1\}$, requiring avoidance of $q_2$. 
We obtain the invariant region $\bar{I}$ where $q_0$ is associated with $[9,100]$, 
$q_1$ with $[0,10]$, and $q_2$ with the empty set, along with a constant function $\bar{V}_1$.  
We then obtain the value function $\bar{V}_0(x,q_0)$ yielding the upper bound shown in \cref{fig:richonce} for $d = 5$ 
 (and more conservative bounds for lower degrees, not shown), with the remaining components defined as the constant functions $\bar{V}_0(x,q_1) = 0$ and $\bar{V}_0(x,q_2) = 1$. \qed

\end{example}

\begin{table}[]
    \centering
    \caption{Output of our quantitative verification experiments.} 
    \begin{tabular}{l|c|r|r}
    \toprule
    Benchmark & $\omega$-Regular Specification & \makecell[c]{Attained\\Bounds} & Time [s] \\
    \hline
    {\tt Gambler's Ruin} $(d=2,a^x)$ & ${\sf GF}(x = 0)$ &  
    [0.380, 0.671]
    & 
    8.64
    \\  
    {\tt Gambler's Ruin} $(d = 3,a^x)$ 
    & ${\sf GF}(x = 0)$ &  
    [0.545, 0.671]
    & 
    9.87
    \\  
    {\tt Gambler's Ruin} $(d = 4,a^x)$ & ${\sf GF}(x = 0)$ &  
    [0.601, 0.671]
    & 
    24.02
    \\
    {\tt Gambler's Ruin} $(d = 5,a^x)$ & ${\sf GF}(x = 0)$ 
    & 
    [0.621, 0.671]
    & 
    88.17
    \\
    \hline
    {\tt BecomingRichOnce} $(d=2)$
    &
    $ (x \geq 10) ~{\sf U}~(x \geq 100)$
    &
    [0.610, 1.000]
    &
    4.68
    \\
    {\tt BecomingRichOnce} $(d=3)$
    & 
    $ (x \geq 10) ~{\sf U}~(x \geq 100)$
    &
    [0.709, 0.974]
    &
    7.39
    \\
    {\tt BecomingRichOnce} $(d=4)$
    & 
    $ (x \geq 10) ~{\sf U}~(x \geq 100)$
    &
    [0.776, 0.908]
    &
    12.15
    \\
    {\tt BecomingRichOnce} $(d=5)$
    & 
    $ (x \geq 10) ~{\sf U}~(x \geq 100)$
    &
    [0.807, 0.880]
    &
    59.04
    \\
    \hline
    {\tt Reactivity1} $(d=2)$& ${\sf GF}(x \leq 6) \rightarrow {\sf GF}(x \leq 0)$ & [0.166, 0.166] &
    8.18
    \\

    {\tt Reactivity2} $(d=2)$ & ${\sf GF}(x \leq 10) \rightarrow {\sf GF}(x \geq 100)$ & [0.250, 0.250] & 
    8.26
    \\
    \bottomrule
    \end{tabular}
    
    \label{tab:experiments}
\end{table}

\begin{table}[]
    \centering
    \caption{Output of our quantitative control synthesis experiments.} 
    \begin{tabular}{l|c|r|r}
    \toprule
    Benchmark & $\omega$-Regular Specification & \makecell[c]{Target\\Bounds} & Time [s] \\
\hline
    {\tt Gambler's Ruin} $(d=2)$ & ${\sf GF}(x = 0)$ &  
    [0.999, 1.000]
    &
    0.69
    \\
    {\tt Becoming Rich Once} $(d=5)$ & $(x \geq 10)~{\sf U}~(x \geq 100)$ &  
    [0.950, 1.000]
    &
    12.93
    \\
    {\tt Reactivity1} $(d=2)$& ${\sf GF}(x \leq 6) \rightarrow {\sf GF}(x \leq 0)$ & [0.187, 1.000] & 
    4.22
    \\
    {\tt Reactivity2} $(d=2)$ & ${\sf GF}(x \leq 10) \rightarrow {\sf GF}(x \geq 100)$ & [0.542, 1.000] & 4.26 \\
    {\tt RepeatedCoin} $(d=3)$ & {\sf GF}$(x \geq 20)$ & [0.499, 0.501] & 0.87 \\
    \bottomrule
    \end{tabular}
    
    \label{tab:control-experiments}
\end{table}

We apply our algorithm to a number of infinite-state Markov chains and $\omega$-regular specifications (\cref{sec:casestudies}). We consider polynomial templates for which we use 
Handelman's Theorem~\cite[Proposition I.1]{HandelmanOriginalPaper1988} 
to reduce the synthesis problem of Streett supermartingales and a stochastic invariant (\cref{eqn:constraint1,eqn:constraint2,eqn:constraint3,eqn:constraint4,eqn:constraint5}) to a decision problem in the existential theory of the reals; for deriving upper bounds upon {\tt Gambler's Ruin}, we use an exponential template for the stochastic invariant \cite{DBLP:conf/pldi/WangS0CG21}. We solve our decision problems using Z3~\cite{DBLP:journals/cca/JovanovicM12}.

Our quantitative verification experiments in \cref{tab:experiments} seek to compute tight bounds upon the satisfaction probability of a specification. 
Notably, in verification the problem of lower bounding is independent of that of upper bounding the satisfaction probability, and both are solved as separate SMT queries. Our control synthesis experiments in \Cref{tab:control-experiments} seek to compute a control parameter for which the probability of satisfaction lies within given target upper and lower bounds. Notably, in control synthesis the first-order logic formulae corresponding to the upper and lower bound are combined in a conjunction and solved as part of the same SMT query. 

Our encoding exploits the structure of the DSA and the Streett supermartingale drift conditions. We heuristically constrain the stochastic invariant to take value 0 (i.e., satisfaction probability of 1) in sink states identified as surely accepting, and value 1 (i.e., satisfaction probability of 0) and sink states identified as surely rejecting, whereas we synthesise the parameters in every other case.

\section{Related Work}

The problem of quantitative model checking and control under $\omega$-regular specifications for finite state Markov chains (and MDPs) is a classic topic for which scalable and automated tools exist \cite{DBLP:conf/cav/DehnertJK017,DBLP:conf/qest/KatoenZHHJ09,DBLP:conf/cav/KwiatkowskaNP11,DBLP:conf/fm/HahnLSTZ14,DBLP:conf/esop/MoosbruggerBKK21,ProbTermToolAmber2021}. As a consequence of the limit behaviour of Markov chains (cf.\ \cite[Theorem 10.27]{DBLP:books/daglib/0020348} and \cite[Theorem 6.4.5]{durrett2019probability}), the quantitative model-checking question reduces to the computation of probabilities to reach accepting bottom strongly connected components. However, this approach does not apply to infinite state Markov chains, where instead finite abstractions \cite{DBLP:conf/qest/DesharnaisLT08,DBLP:journals/ejcon/0001SRHH12,SA13,DBLP:conf/hybrid/TkachevA13,DBLP:conf/cav/AbateGS24,cav25bisimulation} and proof certificates 
\cite{AutomatedVerificationSynthesisSHS,DBLP:conf/emsoft/DimitrovaM14,prajna2006BarrierCertificatesNonlinear,DBLP:conf/cav/ChakarovS13,DBLP:conf/cav/ChatterjeeGMZ22,DBLP:conf/popl/ChatterjeeNZ17,DBLP:journals/toplas/TakisakaOUH21,DBLP:conf/tacas/ChatterjeeHLZ23,Morgan1996ProofRulesProbabilisticLoops} 
constitute two major approaches.

Considering \textit{almost-sure} satisfaction, proof certificates based on martingale theory have been introduced for the specifications of reachability (cf.\ \cite[Corollary 4.4.8]{douc2018markov} and \cite{DBLP:conf/cav/ChakarovS13,DBLP:conf/cav/TakisakaZWL24,DBLP:journals/pacmpl/Huang0CG19,DBLP:journals/pacmpl/McIverMKK18}), persistence \cite[Section 3.1]{DBLP:conf/tacas/ChakarovVS16}, recurrence \cite[Section 3.2]{DBLP:conf/tacas/ChakarovVS16}, and for  reactivity specifications \cite{DBLP:conf/tacas/DimitrovaFHM16,DBLP:conf/cav/AbateGR24}. 
For quantitative specifications, supermartingale proof rules for stochastic invariance (cf.\ \cite[Theorem 1]{kushner1965stability}, \cite[Corollary 4.4.7]{douc2018markov}, and \cite{DBLP:conf/cav/ChatterjeeGMZ22,DBLP:conf/popl/ChatterjeeNZ17,DBLP:journals/toplas/TakisakaOUH21,DBLP:conf/sas/KatoenMMM10,DBLP:conf/cav/ZhiWLOZ24}), reach-avoidance \cite{DBLP:conf/tacas/ChatterjeeHLZ23}, and persistence \cite{DBLP:journals/csysl/AjeleyeZ24a} have been developed, establishing \textit{lower bounds} on the satisfaction probability. Almost-sure proof certificates and stochastic invariants have been combined (cf.\ \cref{thm:decomp}) to yield proof rules for \textit{upper bounding} the probability of termination (cf.\ \cite[Proposition 4]{ChatterjeeFixpoint2025Tacas}, \cite[Lemma 4.6]{MajumdarPOPL2025}, and \cite[Section 6.1]{DBLP:conf/popl/ChatterjeeNZ17}),
and in the context of cost analysis \cite{DBLP:conf/pldi/Wang0GCQS19,DBLP:conf/pldi/Wang0R21,DBLP:conf/cav/SunFCG23} to prove tail bounds on costs accrued prior to termination (cf.\ \cite[Section 6.3]{DBLP:journals/pacmpl/ChatterjeeGMZ24} and \cite[Theorem 6.8]{DBLP:journals/pacmpl/ChatterjeeGMZ24}). In the context of assertion-violation analysis for almost-surely terminating probabilistic programs, a supermartingale certificate (repulsing supermartingales) for stochastic invariance \cite{DBLP:conf/popl/ChatterjeeNZ17,DBLP:journals/toplas/TakisakaOUH21} is combined with a ranking supermartingale \cite[Section 5.1]{DBLP:conf/pldi/WangS0CG21} to yield upper and lower bounds on the probability of assertion violation. This need to combine supermartingale certificates 
has been interpreted and explained using order theory \cite{DBLP:conf/popl/HasuoSC16,DBLP:journals/pacmpl/HarkKGK20}, also yielding new order-theoretic justifications for classic results in martingale theory \cite[Corollary 4.3(2)]{DBLP:journals/toplas/TakisakaOUH21}.

Our results are reminiscent of prior observations in proof rules for quantitative termination analysis, and more generally weakest pre-expectation bounds, in the analysis of probabilistic programs. The notion of \textit{guard-strengthening} \cite{LowerBoundsPossDivPP} may be applied to derive arbitrarily tight lower bounds on the probability of termination by, in effect, restricting attention to a stochastic invariant (and yielding a new program that enjoys stronger termination probabilities). This same approximation property is established in countable-state MDPs \cite[Lemma 4.6]{MajumdarPOPL2025} with bounded discrete probabilistic choices.
Our \cref{thm:approx-general} shows that this applies not just to reachability, but to the richer class of shift-invariant specifications over general state-space Markov chains, by applying L{\'e}vy's 0-1 Law to the satisfaction probability process. Prior work has exploited the connection between infinite-horizon specifications and L{\'e}vy's 0-1 Law  (cf.\ \cite[Section 3.3]{DBLP:conf/icalp/KieferMSTW20}, \cite[Proposition 4]{DBLP:conf/tacas/DimitrovaFHM16}, \cite[Lemma 2]{DBLP:journals/tcs/Chatterjee07}), but we are the first to connect it with the existence of stochastic invariants.
Furthermore, in the context of termination analysis, prior work has observed that in finite state spaces there exists a stochastic invariant that characterises the quantities of interest without approximation error (cf.\ \cite[Theorem 23]{LowerBoundsPossDivPP} and \cite[Lemma 4.5]{MajumdarPOPL2025}).
Both results may be interpreted by applying \cref{thm:exact-finite} to the case where the specification under study is reachability (\cref{ex:reachAsTerm}). 

Converse results for the existence of proof certificates have been established under further topological assumptions \cite[Chapter 6 and Theorem 9.4.2]{meyn_tweedie_glynn_2009} about the transition kernel (e.g.\ the weak Feller property \cite[Theorem 3.2]{DBLP:journals/csysl/MajumdarSS24}). Under the assumption of a countable state space and bounded discrete probabilistic choices, recent work has introduced a sound and complete supermartingale proof rule for almost-sure termination \cite[Lemma 3.4]{MajumdarPOPL2025}, that is applicable to programs that are almost-surely terminating but not with finite expected time \cite{DBLP:conf/popl/FioritiH15}.

The algorithmic synthesis of supermartingale certificates and stochastic invariants draws upon techniques originally developed for the synthesis of invariants and ranking functions for deterministic systems \cite{DBLP:books/daglib/0080029,DBLP:conf/cav/ColonSS03,DBLP:journals/scp/ErnstPGMPTX07,DBLP:journals/tosem/NguyenKWF14,DBLP:conf/sas/SankaranarayananSM04}. These exploit Farkas' Lemma \cite{AgrawalC018,DBLP:conf/cav/ChakarovS13,DBLP:journals/toplas/ChatterjeeFNH18,DBLP:conf/popl/ChatterjeeNZ17,DBLP:conf/tacas/ColonS01,DBLP:conf/cav/ColonS02} and Positivstellensatz \cite{DBLP:conf/cav/ChatterjeeFG16,ChatterjeeAAAI2025,DBLP:conf/pldi/AsadiC0GM21,DBLP:conf/fm/ChatterjeeGGKZ24,DBLP:conf/cav/ChatterjeeFG16,sankaranarayanan2013LyapunovFunctionSynthesis,she2013DiscoveringPolynomialLyapunov,DBLP:conf/cdc/Papachristodoulou02} results, including Handelman's theorem \cite{HandelmanOriginalPaper1988,DBLP:journals/pacmpl/ChatterjeeGNZ24,DBLP:conf/pldi/ZikelicCBR22,PLDI24HandelmanBayesian} which yields a linear decision problem in certain cases. These reduce the problem of constructing a proof certificate to that of solving a problem in quantifier-free nonlinear real arithmetic, and under further assumptions (including the provision of invariants a priori and for autonomous systems) to a linear program. Beyond one-shot synthesis procedures, methods based on counterexample-guided inductive synthesis \cite{DBLP:conf/tacas/BatzCJKKM23} and certificate learning have been proposed \cite{DBLP:conf/aaai/ZikelicLHC23,DBLP:conf/atva/AnsaripourCHLZ23,DBLP:conf/nips/ZikelicLVCH23,DBLP:conf/tacas/ChatterjeeHLZ23,DBLP:conf/aaai/LechnerZCH22,ctsm,neuralmc,DBLP:conf/sigsoft/GiacobbeKP22}.

\section{Conclusion}

Our result shows that, to bound the probability of a shift-invariant specification from below, 
it suffices to present a stochastic invariant together with an almost-sure certificate conditional to this invariant.
It additionally shows the 
necessary existence of appropriate invariants, bounding the probability with arbitrary approximation gap in the general case, and with no error in the finite case. 

Leveraging our result, we have introduced the first quantitative supermartingale certificates for $\omega$-regular specifications, 
encompassing safety, reachability, reach-avoidance and LTL properties. 
Our new quantitative $\omega$-regular certificates are amenable to algorithmic synthesis using symbolic procedures (e.g., polynomials Positivstellensatz), and are additionally prone to future extensions towards 
machine learning techniques~\cite{DBLP:conf/concur/AbateEGPR23}. 
Our approach provides lower and upper bounds on the probability of satisfaction of these properties and readily extends to automated control synthesis with parametrised control policies. 

Our decomposition into stochastic invariants and almost-sure certificates provides the basis for the future development of further quantitative certificates, restricting the focus on (1) proving shift invariance of the specification under study and (2) defining a proof rule for its almost-sure satisfaction. Our converse results guarantee completeness relative to the adopted proof rule for almost-sure acceptance and the adopted algorithm for their automated construction. Our work lays the foundations for developing new model checking, control synthesis and policy learning algorithms with quantitative formal guarantees. 

\subsubsection*{\ackname}
This work was funded in part by the Advanced Research + Invention Agency (ARIA) under the Safeguarded AI programme.

\bibliographystyle{splncs04}
\bibliography{main}

\appendix
\section{Proof of \cref{thm:language-levy}}

We define the $n^\text{th}$ iterate of the function $\theta : \Omega \to \Omega$ by $\theta^{0}(\omega) = \omega$, and $\theta^{n+1}(\omega) = \theta(\theta^n(\omega))$ for all $\omega \in \Omega$ and $n \in \Nat$.
We use the notation $\mathbf{1}\{ \cdot \}$ for the indicator function of an event.
We say that two events $A \in \mathcal{F}$ and $B \in \mathcal{F}$ are \textit{equivalent up to a } $\Prob_\mu$ \textit{-null set} (or simply \textit{equivalent}, for short) if 
\begin{equation}
	\Prob_\mu(A \cap B^{\sf c}) = 0 \qquad \text{ and } \qquad \Prob_\mu(B \cap A^{\sf c}) = 0,
\end{equation}
or in other words, the symmetric difference of the events $A$ and $B$ has probability zero: $\Prob_\mu\left( (A \cap B^{\sf c}) \cup (B \cap A^{\sf c}) \right) = 0$.

\begin{lemma}\label{lem:multi-step-shift}
Suppose $L \in \mathcal{F}$ is shift-invariant. Then
\begin{equation}
	\omega \in L \Longleftrightarrow \theta^n(\omega) \in L,
\end{equation}
for all $n \in \Nat$ and $\omega \in \Omega$.
\end{lemma}
\begin{proof}
We establish this by induction on $n \in \Nat$.
\begin{itemize}
	\item \textbf{Base case}: Holds trivially because $\theta^0(\omega) = \omega$.
	\item \textbf{Inductive step}: we prove 
	\begin{equation}\label{eqn:ind-step-goal}
		\forall \omega \in \Omega \colon \omega \in L \Longleftrightarrow \theta^{n+1}(\omega) \in L
	\end{equation}
	under the assumption that 
	\begin{equation}
		\forall \omega \in \Omega \colon \omega \in L \Longleftrightarrow \theta^n(\omega) \in L.
	\end{equation}
	
	First, we expand \cref{eqn:shift-invariant} to obtain that 
	\begin{equation}\label{eqn:expanded-shift-invariant}
		\forall \omega \in \Omega \colon \omega \in L \Longleftrightarrow \theta(\omega) \in L
	\end{equation}
	Considering an arbitrary $\omega \in \Omega$:
	\begin{align}
		\omega \in L &\Longleftrightarrow \theta^n (\omega) \in L &\text{by I.H.,}\\
		&\Longleftrightarrow \theta(\theta^n (\omega)) \in L &\text{by \cref{eqn:expanded-shift-invariant},}\\
		&\Longleftrightarrow \theta^{n + 1}(\omega) \in L &\text{by definition of $\theta^{n+1}$,}
	\end{align}
	which proves \cref{eqn:ind-step-goal}.
	\end{itemize}\qed
    
\end{proof}

Restating \cref{lem:multi-step-shift}, the events $\{ \mc \in L \}$ and the event $\{ \mc \circ \theta^n \in L \}$ are equivalent for all $n \in \Nat$, and hence the equality
\begin{equation}\label{eqn:equality-shifted-indicators}
	\ind\{ \mc \in L \} = \ind\{ \mc \circ \theta^n \in L \}
\end{equation}
holds $\Prob_\mu$-almost surely, for all $n \in \Nat$. 

By the time-homogeneous Markov property \cref{eqn:time-hom-m-p} applied to the random variable $H = \ind \{\mc \in L \}$:
\begin{align}\label{eqn:invoke-markov-property}
	\Expect_\mu [ \ind \{ \mc \circ \theta^n \in L \} \mid {\cal F}^\mc_n ] 
	&= \Expect_{\Phi_n}[ \ind \{ \mc \in L \} ]\\
	&= \Prob_{\Phi_n} ( \mc \in L ).\label{eqn:invoke-markov-property-2}
\end{align}

Note that we used the fact that for any event $A \in \mathcal{F}$, $\Expect_{\Phi_n} [ \ind \{ \mc \in A \} ] = \Prob_{\Phi_n} ( \mc \in A )$, for all $n \in \Nat$.

By \cref{eqn:equality-shifted-indicators} and \cref{eqn:invoke-markov-property-2}:
\begin{align}\label{eqn:prob-is-cexp}
	\Expect_\mu [ \ind \{ \mc \in L \} \mid {\cal F}^\mc_n ]
	&= \Prob_{\Phi_n} ( \mc \in L ).
\end{align}
for all $n \in \Nat$.

\begin{lemma}\label{lem:martingale-prob-sequence}
Suppose $L \in \mathcal{F}$ is shift-invariant.	Then, 
the sequence
\begin{equation}\label{eqn:sequence-probabilities-martingale}
	\Prob_{\Phi_0}(\mc \in L),
	\Prob_{\Phi_1}(\mc \in L),
	\ldots,
	\Prob_{\Phi_n}(\mc \in L),
	\ldots
\end{equation}
of random variables $\Prob_{\Phi_n}(\mc \in L): \Omega \to [0, 1]$ is a non-negative martingale adapted to the filtration ${\cal F}_n^\Phi$.
\end{lemma}
\begin{proof}

We show that the process is a non-negative martingale \cite[Section 5.2, p.232]{durrett2019probability}, namely, that it is (i) non-negative, (ii) integrable, (iii) adapted, and (iv) satisfies the martingale property:
\begin{enumerate}
	\item \textbf{Non-negativity.} This holds because for all $\omega \in \Omega$ the value $\Prob_{\Phi_n(\omega)}(\mc \in L)$ is a probability, and therefore non-negative.
	\item \textbf{Integrability.} This holds because for each $n \in \Nat$, the random variable $\Prob_{\Phi_n}(\mc \in L)$ is upper bounded by 1, meaning its absolute value is finite in expectation: $\Expect_\mu [ ~| \Prob_{\Phi_n}(\mc \in L) |~] \leq 1 < \infty$ for all $n \in \Nat$.
	\item \textbf{Adaptedness.} This holds because the function $s \mapsto \Prob_s (\mc \in L) : S \to \Real$ is $\Sigma / \mathcal{B}(\Real)$-measurable \cite[Proposition 5.2.2(i), p.104]{douc2018markov}, and $\Phi_n : \Omega \to S$ is ${\cal F}_n^\Phi / \Sigma$-measurable, therefore the composition $\omega \mapsto \Prob_{\Phi_n(\omega)} (\mc \in L) : \Omega \to \Real$ is ${\cal F}_n^\Phi / {\cal B}(\Real)$-measurable, as required.
	\item \textbf{Martingale property.} Namely, we must show that 
	\begin{equation}
		\Expect_\mu [ \Prob_{\Phi_{n + 1}}(\mc \in L) 
		\mid {\cal F}_n^\mc
		 ] = \Prob_{\Phi_n}(\mc \in L)
	\end{equation}
	holds $\Prob_\mu$-almost surely.
	This follows by applying the \textit{tower property} of conditional expectations \cite[Theorem 5.1.6, Section 5.1, p.228]{durrett2019probability} applied to the $\sigma$-algebras ${\cal F}_n^\mc \subseteq {\cal F}_{n + 1}^\mc$:
	\begin{align}
		&\Expect_\mu [ 
		\Prob_{\Phi_{n + 1}}(\mc \in L) 
		\mid 
		{\cal F}_n^\mc 
		]\\
		&= 
		\Expect_\mu [ 
		\Expect_\mu [
		\ind \{ \mc \in L \}		
		\mid 
		{\cal F}_{n + 1}^\Phi
		]
		\mid 
		{\cal F}_n^\Phi
		] &\text{by \cref{eqn:prob-is-cexp},}\\
		&= 
		\Expect_\mu [ 
		\ind \{ \mc \in L \} \mid 
		{\cal F}_n^\Phi 
		] &\text{by tower property,} \\
		&= \Prob_{\Phi_n} ( \mc \in L ) &\text{by \cref{eqn:prob-is-cexp}}.
	\end{align}
\end{enumerate}
This establishes that the sequence \cref{eqn:sequence-probabilities-martingale} is a martingale adapted to ${\cal F}_n^\Phi$.
\end{proof}

\begin{lemma}
	Suppose $L \in \mathcal{F}$ is shift-invariant. Then
	\begin{equation}\label{eqn:goal1}
		\Prob_\mu 
		\left( 
		\left(\forall n \in \Nat \colon \Prob_{\Phi_n}(\mc \in L) > 0\right)
		\lor 
		\left(
		\lim_{n \to \infty} 
		\Prob_{\Phi_n}(\mc \in L) = 0
		\right)
		\right) = 1
	\end{equation}
\end{lemma}
\begin{proof}

We first show that 
\begin{align}\label{eqn:absorb-step-goal}
	\Prob_\mu 
	(
	\Prob_{\Phi_n}(\mc \in L) = 0 
	\land
	\Prob_{\Phi_{n + 1}}(\mc \in L) > 0
	) = 0
\end{align}
for all $n \in \Nat$.

We recall the definition of the conditional expectation $\Expect_\mu [ \Prob_{\Phi_{n + 1}} (\mc \in L) \mid {\cal F}_n^\Phi ]$ as an ${\cal F}_n^\Phi$-measurable random variable satisfying the following \textit{averaging} property \cite[Section 5.1, p.221]{durrett2019probability}:
\begin{equation}\label{eqn:c-expect-average}
	\int_{\omega \in A} 
	\Prob_{\Phi_{n + 1}(\omega)}(\mc \in L) ~\Prob_\mu(\mathrm{d}\omega)
	=
	\int_{\omega \in A}
	\Expect_\mu [ 
	\Prob_{\Phi_{n + 1}} (\mc \in L)
	\mid 
	{\cal F}_n^\Phi
	](\omega) ~\Prob_\mu(\mathrm{d}\omega),
\end{equation}
for all $A \in {\cal F}_n^\Phi$.

By the fact that the sequence $\Prob_{\Phi_n}(\mc \in L)$ is a martingale (\Cref{lem:martingale-prob-sequence}), we rewrite this as:
\begin{equation}\label{eqn:c-expect-average-substituted}
	\forall A \in {\cal F}_n^\Phi \colon
	\int_{\omega \in A} 
	\Prob_{\Phi_{n + 1}(\omega)}(\mc \in L) ~\Prob_\mu(\mathrm{d}\omega)
	=
	\int_{\omega \in A}
	\Prob_{{\Phi_n}(\omega)} (\mc \in L)
	~\Prob_\mu(\mathrm{d}\omega).
\end{equation}
where we replaced $	\Expect_\mu [ 
	\Prob_{\Phi_{n + 1}} (\mc \in L)
	\mid 
	{\cal F}_n^\Phi
	](\omega)$ by $\Prob_{\Phi_n(\omega)} (\mc \in L)$ using \cref{lem:martingale-prob-sequence}.
	
We note that the event 
\begin{equation}\label{eqn:zero-set-of-prob}
	\left\{ \omega \in \Omega \colon \Prob_{\Phi_n(\omega)}(\mc \in L) = 0 \right\}
\end{equation}
is an element of the $\sigma$-algebra ${\cal F}_n^\Phi$ as $\Prob_{\Phi_n}(\Phi \in L)$ is ${\cal F}_n^\Phi / {\cal B}(\Real)$-measurable.

Instantiating \cref{eqn:c-expect-average-substituted} by setting $A$ to the event \cref{eqn:zero-set-of-prob}, we obtain that 
\begin{equation}
	\int_{\omega : \Prob_{\Phi_n(\omega)}(\mc \in L) = 0}
	\Prob_{\Phi_{n + 1}(\omega)}(\mc \in L)~
	\Prob_\mu(\mathrm{d}\omega) = 0,
\end{equation}
which implies (e.g., by \cite[Lemma 26(iv), p.33]{pollard_2001}) that 
\begin{equation}\label{eqn:one-step-absorption}
	\Prob_\mu \left( 
	\Prob_{\Phi_n}(\mc \in L) = 0 
	\land 
	\Prob_{\Phi_{n + 1}}(\mc \in L) > 0
	\right) 
	= 0.
\end{equation}
This establishes \cref{eqn:absorb-step-goal}.

We then prove that the event 
\begin{equation}\label{eqn:goal-2-k}
	\Prob_{\Phi_n}(\mc \in L) = 0
	\land 
	\Prob_{\Phi_{n + k + 1}}(\mc \in L) > 0
\end{equation}
has $\Prob_\mu$-measure zero for all $n,k \in \Nat$ by induction on $k \in \Nat$ for an arbitrary $n \in \Nat$.
\begin{itemize}
	\item \textbf{Base case.} We must show that 
	\begin{equation}
	\Prob_\mu \left(
	\Prob_{\Phi_n}(\mc \in L) = 0
	\land 
	\Prob_{\Phi_{n + 1}}(\mc \in L) > 0
	\right) = 0
	\end{equation}
	which is immediate from \cref{eqn:absorb-step-goal}.
	
	\item \textbf{Inductive step.} We must prove 
	\begin{equation}\label{eqn:goal-inductive-step}
		\Prob_\mu \left(
		\Prob_{\Phi_n}(\mc \in L) = 0
		\land 
		\Prob_{\Phi_{n + k + 2}}(\mc \in L) > 0
		\right) = 0
	\end{equation}
	under the assumption
	\begin{equation}\label{eqn:absorb-induc-hyp}
		\Prob_\mu \left(
		\Prob_{\Phi_n}(\mc \in L) = 0
		\land 
		\Prob_{\Phi_{n + k + 1}}(\mc \in L) > 0
		\right) = 0.
	\end{equation}
	
	By instantiating \cref{eqn:absorb-step-goal} by replacing $n$ with $n + k + 1$, we obtain:
	\begin{equation}\label{eqn:absorb-step-shifted}
		\Prob_\mu \left(
		\Prob_{\Phi_{n+k + 1}}(\mc \in L) = 0
		\land 
		\Prob_{\Phi_{n + k + 2}}(\mc \in L) > 0
		\right) = 0.
	\end{equation}
	
	From \cref{eqn:absorb-step-shifted} we take a conjunction with event $ \Prob_{\Phi_n} (\mc \in L) = 0 $ to conclude
	\begin{equation}\label{eqn:step-case-zero}
		\Prob_\mu \left(
		\Prob_{\Phi_n} (\mc \in L) = 0 
		\land
		\Prob_{\Phi_{n+k + 1}}(\mc \in L) = 0
		\land 
		\Prob_{\Phi_{n + k + 2}}(\mc \in L) > 0
		\right) = 0,
	\end{equation} 
	and separately, starting with the inductive hypothesis \cref{eqn:absorb-induc-hyp} we take a conjunction with the event $\Prob_{\Phi_{n + k + 2}}(\mc \in L)  > 0$ to conclude:
	\begin{equation}\label{eqn:step-case-nonzero}
		\Prob_\mu \left(
		\Prob_{\Phi_n} (\mc \in L) = 0 
		\land
		\Prob_{\Phi_{n + k + 1}}(\mc \in L) > 0
		\land 
		\Prob_{\Phi_{n + k + 2}}(\mc \in L) > 0
		\right) = 0,
	\end{equation}
	where in both cases we used the fact that a conjunction of an event with a zero probability event yields an event that has zero probability.
	
	By the fact that the event 
	\begin{equation}
		\Prob_{ \Phi_{n + k + 1}}(\mc \in L) = 0 
		\lor 
		\Prob_{ \Phi_{n + k + 1}}(\mc \in L) > 0
	\end{equation}
	occurs $\Prob_\mu$-almost surely, and taking the union of the disjoint events referred to by \cref{eqn:step-case-nonzero,eqn:step-case-zero}, we obtain that
    \begin{equation}
        \Prob_\mu \left(
		\Prob_{\Phi_n} (\mc \in L) = 0 
		\land 
		\Prob_{\Phi_{n + k + 2}}(\mc \in L) > 0
		\right) = 0,
    \end{equation}
    which proves \cref{eqn:goal-inductive-step}, and thereby, by induction, establishes \cref{eqn:goal-2-k}.	
\end{itemize}

We may restate \cref{eqn:goal-2-k} as:
\begin{equation}\label{eqn:restate-1}
\forall n \in \Nat, \forall m > n \colon
    \Prob_\mu\left( 
    \Prob_{\Phi_n}(\mc \in L) = 0 
    \land 
    \Prob_{\Phi_m}(\mc \in L) > 0
    \right) 
    = 0.
\end{equation}

By applying to \cref{eqn:restate-1} the fact that a countable union of probability zero events has probability zero, we obtain:
\begin{equation}\label{eqn:intermediary-1}
    \forall n \in \Nat \colon
    \Prob_\mu \left(
    \Prob_{\Phi_n}(\mc \in L) = 0 
    \land
    \left( 
    \exists m > n \colon
    \Prob_{\Phi_m}(\mc \in L) > 0
    \right)
    \right) = 0.
\end{equation}

Taking the complement of the above event mentioned in \cref{eqn:intermediary-1} for each $n \in \Nat$:
\begin{equation}\label{eqn:intermediary-2}
    \forall n \in \Nat \colon 
    \Prob_\mu \left( 
    \Prob_{\Phi_n}(\mc \in L) > 0 
    \lor
    \forall m > n \colon 
    \Prob_{\Phi_m}(\mc \in L) = 0
    \right)
    = 1.
\end{equation}

We note that for all $n \in \Nat$, the event 
\begin{equation}
    \left\{ 
    \forall m > n \colon 
    \Prob_{\Phi_m}(\mc \in L)  = 0
    \right\}
\end{equation}
is a subset of the event 
\begin{equation}
    \left\{ 
    \lim_{n \to \infty} 
    \Prob_{\Phi_n}(\mc \in L) = 0
    \right\}.
\end{equation}

Therefore, we rewrite \cref{eqn:intermediary-2} into:
\begin{equation}\label{eqn:intermediary-3}
    \forall n \in \Nat \colon
    \Prob_\mu \left( 
    \Prob_{\Phi_n}(\mc \in L) > 0 
    \lor
    \lim_{n \to \infty}
    \Prob_{\Phi_n}(\mc \in L) = 0
    \right)
    =
    1.
\end{equation}

Since a countable intersection of probability 1 events has probability 1, we conclude from \cref{eqn:intermediary-3} that:
\begin{equation}\label{eqn:intermediary-4}
    \Prob_\mu\left(
    \left( 
    \forall n \in \Nat \mathpunct.
    \Prob_{\Phi_n}(\mc \in L) > 0
    \right)
    \lor 
    \left( 
    \lim_{n \to \infty}
    \Prob_{\Phi_n}(\mc \in L)  = 0
    \right)
    \right)
    = 1
\end{equation}
as desired in \cref{eqn:goal1}.\qed
\end{proof}

By complementing \cref{eqn:intermediary-4}, we obtain
\begin{equation}\label{eqn:intermediary-5}
    \Prob_\mu\left( 
    \left(
    \exists n \in \Nat \colon 
    \Prob_{\Phi_n}(\mc \in L) = 0 \right)
    \land 
    (\Omega \smallsetminus
    \lim_{n \to \infty}
    \Prob_{\Phi_n}(\mc \in L) = 0)
    \right)
    = 0.
\end{equation}
By L{\'evy}'s 0-1 Law \cite[Theorem 5.5.8]{durrett2019probability}, the event
\begin{equation}
    \Omega \smallsetminus \lim_{n \to \infty} \Prob_{\Phi_n}(\mc \in L) = 0,
\end{equation}
is $\Prob_\mu$-equivalent to the event
\begin{equation}
    \lim_{n \to \infty} \Prob_{\Phi_n}(\mc \in L)  = 1.
\end{equation}
By this observation and \cref{eqn:intermediary-5} we conclude:
\begin{equation}\label{eqn:intermediary-6}
    \Prob_\mu\left(
    \exists n \in \Nat \colon
    \Prob_{\Phi_n}(\mc \in L) = 0
    \land 
    \lim_{n \to \infty}
    \Prob_{\Phi_n}(\mc \in L) = 1
    \right)
    = 0.
\end{equation}
Since
\begin{equation}\label{eqn:decomposition-omega-ltp}
    \Omega = \left\{ 
    \exists n \in \Nat \colon 
    \Prob_{\Phi_n}(\mc \in L) = 0
    \right\}
    \cup \left\{ 
    \forall n \in \Nat \colon \Prob_{\Phi_n}(\mc \in L) > 0
    \right\}
\end{equation}
we conclude by \cref{eqn:intermediary-6} and the law of total probability (using \cref{eqn:decomposition-omega-ltp}) that:
\begin{equation}
\begin{aligned}\label{eqn:intermed1}
    &\Prob_\mu\left(
    \lim_{n \to \infty}
    \Prob_{\Phi_n}(\mc \in L) = 1\right)\\
    &\qquad = 
    \Prob_\mu\left(
    \left(
    \lim_{n \to \infty}
    \Prob_{\Phi_n}(\mc \in L) = 1
    \right)
    \land
    \left(
    \forall n \in \Nat \colon \Prob_{\Phi_n}(\mc \in L) > 0
    \right)
    \right).
\end{aligned}
\end{equation}

\begin{lemma}\label{eqn:replace-safe-with-inf}
The event 
\begin{equation}
    \left(
    \lim_{n \to \infty}
    \Prob_{\Phi_n}(\mc \in L) = 1
    \right)
    \land
    \left(
    \forall n \in \Nat \colon \Prob_{\Phi_n}(\mc \in L) > 0
    \right)
\end{equation}
is equivalent to the event 
\begin{equation}
    \inf_n \Prob_{\Phi_n}(\mc \in L) > 0
    \land 
    \lim_{n \to \infty}
    \Prob_{\Phi_n}(\mc \in L) = 1.
\end{equation}
\end{lemma}
\begin{proof} 

We prove in fact that the two events are identical in the sense that they are the same subset of $\Omega$:

$(\subseteq)$ Supposing that 
    \begin{equation}\label{eqn:premise-limit-1}
        \lim_{n \to \infty} \Prob_{\Phi_n(\omega)}(\mc \in L) = 1
    \end{equation} and 
    \begin{equation}\label{eqn:premise-g-pos}
    \forall n \in \Nat \colon \Prob_{\Phi_n(\omega)}(\mc \in L) > 0,
    \end{equation}
    and expanding the definition of \cref{eqn:premise-limit-1} we obtain
    \begin{equation}\label{eqn:limit-with-eps}
        \forall \epsilon > 0~
        \exists m~
        \forall n \geq m \colon
        \Prob_{\Phi_n(\omega)}(\mc \in L) \geq 1 - \epsilon.
    \end{equation}
    \Cref{eqn:limit-with-eps} implies, by substituting $\epsilon = \frac{1}{2}$ (although any positive value would be sufficient):
    \begin{equation}\label{eqn:expand-1}
        \exists m ~
        \forall n \geq m 
        \colon 
        \Prob_{\Phi_n(\omega)}(\mc \in L) \geq \frac{1}{2}
    \end{equation}
    and taking $m_0$ as the witness for $m$ in \cref{eqn:expand-1} which in combination with \cref{eqn:premise-g-pos}, means that 
    \begin{equation}
        \inf_{n} \Prob_{\Phi_n(\omega)}(\mc \in L)
        \geq \min\left( \Prob_{\Phi_0(\omega)}(\mc \in L),\ldots, 
        \Prob_{\Phi_{m_0}(\omega)}(\mc \in L), 
        \frac{1}{2}
        \right),
    \end{equation}
    but this lower bound is the minimum of a finite number of strictly positive quantities, and is therefore itself strictly positive.
    Hence we have established that 
    \begin{equation}
        \inf_{n} \Prob_{\Phi_n(\omega)}(\mc \in L) > 0.
    \end{equation}

    $(\supseteq)$ This follows immediately from the observation that if the infimum of a sequence is strictly positive, then all terms in the sequence must be strictly positive.
\qed
\end{proof}

Finally, since the event 
\begin{equation}
    \inf_{n} \Prob_{\Phi_n(\omega)}(\mc \in L) > 0
    \land 
    \lim_{n \to \infty}
    \Prob_{\Phi_n}(\mc \in L) = 0
\end{equation}
is empty, combining \cref{eqn:intermediary-6,eqn:intermed1,eqn:replace-safe-with-inf}, and applying the law of total probability, we obtain that:
\begin{equation}
    \Prob_\mu
    \left( 
    \lim_{n \to \infty} 
    \Prob_{\Phi_n}(\mc \in L) = 1
    \right)
    = 
    \Prob_\mu\left(
    \inf_n \Prob_{\Phi_n}(\mc \in L) > 0\right), 
\end{equation}
which completes the proof of \cref{thm:language-levy}, since (by L{\'e}vy 0-1 Law) the events $\mc \in L$ and the event $\lim_{n \to \infty} 
    \Prob_{\Phi_n}(\mc \in L) = 1$ are $\Prob_\mu$-equivalent.

\section{Proof of \cref{thm:approx-general}}

The event 
\begin{equation}
    \{ \mc \in L \}
\end{equation}
is $\Prob_\mu$-equivalent to the event 
\begin{equation}\label{eqn:inf-gt-0}
    \inf_n \Prob_{\Phi_n}(\mc \in L) > 0
\end{equation}
as established by \cref{thm:language-levy}.

The event \cref{eqn:inf-gt-0} is equal to the following event
\begin{equation}
    \bigcup_{k = 0}^{\infty}
    \left\{ 
    \inf_n \Prob_{\Phi_n}(\mc \in L) > \frac{1}{k+1} 
    \right\}
\end{equation}
which is a countable increasing union. 
Applying the Monotone Convergence Theorem (\cite[Theorem 12, p.26]{pollard_2001} and \cite[Theorem 2.59]{axler2019measure}), the probability of the event 
\begin{equation}\label{eqn:approximant-k}
    \left\{ 
    \inf_n \Prob_{\Phi_n}(\mc \in L) > \frac{1}{k+1} 
    \right\}
\end{equation}
converges monotonically to the probability of the event \cref{eqn:inf-gt-0}, as $k \rightarrow \infty$. 
Expanding definitions, this means that for every $\epsilon > 0$ there exists $k_0 \in \Nat$ such that for all $k \geq k_0$, we have 
\begin{equation}\label{eqn:lower1}
    \Prob_\mu\left( 
    \inf_n \Prob_{\Phi_n}(\mc \in L) \geq \frac{1}{k+1}
    \right) \geq 
    \Prob_\mu\left( 
    \inf_n \Prob_{\Phi_n}(\mc \in L) > 0
    \right) - \epsilon 
\end{equation}
and furthermore, since the event \cref{eqn:approximant-k} is a subset of the event \cref{eqn:inf-gt-0} for all $k \in \Nat$, it follows that 
\begin{equation}\label{eqn:upper1}
    \Prob_\mu\left( 
    \inf_n \Prob_{\Phi_n}(\mc \in L) \geq \frac{1}{k+1}
    \right) \leq 
    \Prob_\mu\left( 
    \inf_n \Prob_{\Phi_n}(\mc \in L) > 0
    \right)
\end{equation}

Combining \cref{eqn:lower1,eqn:upper1} we obtain that for all $\epsilon > 0$, there exists $k_0 \in \Nat$ for which:
\begin{equation}
\label{eqn:sandwich-inf-statement-1}
    \Prob_\mu\left( 
    \inf_n \Prob_{\Phi_n}(\mc \in L) > 0
    \right) - \epsilon
    \leq 
    \Prob_\mu\left( 
    \inf_n \Prob_{\Phi_n}(\mc \in L) \geq \frac{1}{k_0+1}
    \right)
\end{equation}
and
\begin{equation}\label{eqn:sandwich-inf-statement-2}
    \Prob_\mu\left( 
    \inf_n \Prob_{\Phi_n}(\mc \in L) \geq \frac{1}{k_0+1}
    \right) \leq 
    \Prob_\mu\left( 
    \inf_n \Prob_{\Phi_n}(\mc \in L) > 0
    \right).
\end{equation}

Define the set $I \in \Sigma$ by
\begin{equation}
    I = \left\{ 
    s \in S \colon 
    \Prob_{\delta_s}(\mc \in L) \geq \frac{1}{k_0 + 1}
    \right\}.
\end{equation}
Then, we note that for all $\omega \in \Omega$:
\begin{align}
    &\inf_n \Prob_{\Phi_n(\omega)}(\mc \in L) \geq \frac{1}{k_0+1}\\
    &\Longleftrightarrow
    \forall n \in \Nat \colon 
    \Prob_{\Phi_n(\omega)}(\mc \in L) \geq \frac{1}{k_0+1}\\
    &\Longleftrightarrow
    \forall n \in \Nat \colon 
    \Phi_n(\omega) \in I \\
    &\Longleftrightarrow \mc(\omega) \in I^\omega.
\end{align}

Using this fact, and the equality established by \cref{thm:language-levy} we may rewrite \cref{eqn:sandwich-inf-statement-1,eqn:sandwich-inf-statement-2} into: 
\begin{equation}
\forall \epsilon > 0,
\exists I \in \Sigma \colon
    \Prob_\mu\left( 
    \mc \in L
    \right) - \epsilon
    \leq 
    \Prob_\mu\left( 
    \mc \in I^\omega
    \right) \leq 
    \Prob_\mu\left( 
    \mc \in L
    \right).
\end{equation}

The fact that the event \cref{eqn:approximant-k} is a subset of the event \cref{eqn:inf-gt-0} ensures ${\sf P}_\mu(\mc \in I^\omega \land \mc \notin L) = 0$.

\section{Proof of \cref{thm:exact-finite}}

Supposing that $S$ is finite, define 
\begin{equation}
    I = \left\{ 
    s \in S \colon \Prob_{\delta_s}(\mc \in L) \geq 
    \min_{s \in S \colon \Prob_{\delta_s}(\mc \in L) > 0} \Prob_{\delta_s}(\mc \in L) 
    \right\} 
\end{equation}

Then, we note that for all $\omega \in \Omega$:
\begin{align}
    &\inf_n \Prob_{\Phi_n(\omega)}(\mc \in L) > 0\\
    &\Longleftrightarrow 
    \inf_n \Prob_{\Phi_n(\omega)}(\mc \in L) \geq \min_{s \in S : \Prob_{\delta_s}(\mc \in L) > 0} \Prob_{\delta_s}(\mc \in L)\\
    &\Longleftrightarrow 
    \forall n \in \Nat \colon 
    \Prob_{\Phi_n(\omega)}(\mc \in L) \geq \min_{s \in S : \Prob_{\delta_s}(\mc \in L) > 0} \Prob_{\delta_s}(\mc \in L)\\
    &\Longleftrightarrow
    \forall n \in \Nat \colon 
    \Phi_n(\omega) \in I\\
    &\Longleftrightarrow
    \mc(\omega) \in I^\omega
\end{align}
where $\min_{s \in S : \Prob_{\delta_s}(\mc \in L) > 0}(\mc \in L)$ exists and is strictly greater than zero, being a minimum of a finite number of strictly positive values.

\section{Proof of \cref{thm:decomp}}
 Suppose 
    \begin{equation}
    {\sf P}_\mu(\mathbf{\Phi} \in I^\omega) \geq p \land
    {\sf P}_\mu (\mathbf{\Phi} \in L \mid \mathbf{\Phi} \in I^\omega) = 1
    \end{equation}
    then by expanding the definition of conditional expectation
    \begin{equation}
    {\sf P}_\mu(\mathbf{\Phi} \in I^\omega) \geq p \land
    \frac{{\sf P}_\mu (\mathbf{\Phi} \in L \land \mathbf{\Phi} \in I^\omega)}{
    {\sf P}_\mu (\mathbf{\Phi} \in I^\omega)
    } = 1
    \end{equation}
    Then, by the law of total probability we have 
    \begin{equation}\label{eqn:lawTotalProb}
        {\sf P}_\mu(\mathbf{\Phi} \in L) = {\sf P}_\mu({\bf \Phi} \in L \land {\bf \Phi} \in I^\omega) + {\sf P}_\mu({\bf \Phi} \in L \land {\bf \Phi} \notin I^\omega)
    \end{equation}
    and therefore we have
    \begin{equation}
    {\sf P}_\mu(\mathbf{\Phi} \in I^\omega) \geq p \land
    \frac{{\sf P}_\mu(\mathbf{\Phi} \in L) -  {\sf P}_\mu (\mathbf{\Phi} \in L \land \mathbf{\Phi} \notin I^\omega)}{
    {\sf P}_\mu (\mathbf{\Phi} \in I^\omega)
    } = 1
    \end{equation}
Multiplying the denominator:
\begin{equation}
    {\sf P}_\mu(\mathbf{\Phi} \in I^\omega) \geq p \land
    {\sf P}_\mu(\mathbf{\Phi} \in L) -  {\sf P}_\mu (\mathbf{\Phi} \in L \land \mathbf{\Phi} \notin I^\omega)
     = {\sf P}_\mu (\mathbf{\Phi} \in I^\omega)
    \end{equation}
and therefore
\begin{equation}
    {\sf P}_\mu(\mathbf{\Phi} \in I^\omega) \geq p \land
    {\sf P}_\mu(\mathbf{\Phi} \in L) 
     = {\sf P}_\mu (\mathbf{\Phi} \in I^\omega) 
    + 
    {\sf P}_\mu (\mathbf{\Phi} \in L \land \mathbf{\Phi} \notin I^\omega)
    \end{equation}
This implies that 
\begin{equation}
    {\sf P}_\mu(\mathbf{\Phi} \in L) 
     = {\sf P}_\mu (\mathbf{\Phi} \in I^\omega) 
    + 
    {\sf P}_\mu (\mathbf{\Phi} \in L \land \mathbf{\Phi} \notin I^\omega) \geq {\sf P}_\mu(\mathbf{\Phi} \in I^\omega) \geq p.
\end{equation}

\section{Proof of \cref{thm:epsilon-complete}}
Let $\epsilon > 0$ and suppose
\begin{equation}\label{eqn:premise}
    {\sf P}_\mu(\mathbf{\Phi} \in L) \geq p
\end{equation}
From \cref{eqn:premise} and the law of total probability \cref{eqn:lawTotalProb}:
\begin{equation}
    {\sf P}_\mu({\bf \Phi} \in L \land {\bf \Phi} \in I^\omega) + {\sf P}_\mu({\bf \Phi} \in L \land {\bf \Phi} \notin I^\omega)
    \geq p.
\end{equation}

By \cref{thm:approx-general} there exists $I \in \Sigma$ for which 
    \begin{equation}
    {\sf P}_\mu({\bf \Phi} \in I^\omega \land {\bf \Phi} \notin L) = 0 
    \land 
        {\sf P}({\bf {\Phi}} \in I^\omega) \leq {\sf P}({\bf {\Phi}} \in L) \leq {\sf P}({\bf {\Phi}} \in I^\omega) + \epsilon
    \end{equation}

By adding ${\sf P}_\mu({\bf \Phi} \in I^\omega \land {\bf \Phi} \in L)$ to both sides of the first conjunct:
\begin{equation}
    \begin{aligned}
    &{\sf P}_\mu({\bf \Phi} \in I^\omega \land {\bf \Phi} \in L) + 
    {\sf P}_\mu({\bf \Phi} \in I^\omega \land {\bf \Phi} \notin L) = {\sf P}_\mu({\bf \Phi} \in I^\omega \land {\bf \Phi} \in L)\\
    &\land 
        {\sf P}({\bf {\Phi}} \in I^\omega) \leq {\sf P}({\bf {\Phi}} \in L) \leq {\sf P}({\bf {\Phi}} \in I^\omega) + \epsilon
    \end{aligned}
\end{equation}
By law of total probability \cref{eqn:lawTotalProb} applied to the first conjunct:
\begin{equation}
    {\sf P}_\mu({\bf \Phi} \in I^\omega) = {\sf P}_\mu({\bf \Phi} \in I^\omega \land {\bf \Phi} \in L)
    \land 
        {\sf P}({\bf {\Phi}} \in I^\omega) \leq {\sf P}({\bf {\Phi}} \in L) \leq {\sf P}({\bf {\Phi}} \in I^\omega) + \epsilon
    \end{equation}
Dividing both sides of first conjunct by ${\sf P}_\mu({\bf \Phi} \in I^\omega)$:
\begin{equation}
    1 = \frac{{\sf P}_\mu({\bf \Phi} \in I^\omega \land {\bf \Phi} \in L)}{{\sf P}_\mu({\bf \Phi} \in I^\omega)}
    \land 
        {\sf P}({\bf {\Phi}} \in I^\omega) \leq {\sf P}({\bf {\Phi}} \in L) \leq {\sf P}({\bf {\Phi}} \in I^\omega) + \epsilon
    \end{equation}
Using the definition of ${\sf P}_\mu (\mathbf{\Phi} \in L \mid \mathbf{\Phi} \in I^\omega)$ 
\begin{equation}
    1 = {\sf P}_\mu (\mathbf{\Phi} \in L \mid \mathbf{\Phi} \in I^\omega)
    \land 
        {\sf P}({\bf {\Phi}} \in I^\omega) \leq {\sf P}({\bf {\Phi}} \in L) \leq {\sf P}({\bf {\Phi}} \in I^\omega) + \epsilon
    \end{equation}
Using ${\sf P}({\bf {\Phi}} \in L) \geq p$:
\begin{equation}
    1 = {\sf P}_\mu (\mathbf{\Phi} \in L \mid \mathbf{\Phi} \in I^\omega)
    \land 
        p \leq {\sf P}({\bf {\Phi}} \in L) \leq {\sf P}({\bf {\Phi}} \in I^\omega) + \epsilon
    \end{equation}
and rearranging the inequalities in the second conjunct:
\begin{equation}
    1 = {\sf P}_\mu (\mathbf{\Phi} \in L \mid \mathbf{\Phi} \in I^\omega)
    \land
     {\sf P}({\bf {\Phi}} \in I^\omega) \geq p - \epsilon
    \end{equation}

\section{Proof of \cref{thm:complete}}

Suppose $L \in {\cal F}$ is shift-invariant and $S$ is finite.
By \cref{thm:exact-finite} there exists an $I \in \Sigma$ such that 
\begin{equation}
    {\sf P}_\mu(\mc \in I^\omega \land \mc \notin L ) = 0 \land {\sf P}_\mu( \mc \in I^\omega ) = {\sf P}_\mu(\mc \in L).
\end{equation}
This demonstrates satisfaction of \cref{eqn:pee} with $p = {\sf P}_\mu(\mc \in L)$, so we turn to showing that \cref{eqn:one} holds, starting with
\begin{equation}
    {\sf P}_\mu(\mc \in I^\omega \land \mc \notin L ) = 0.
\end{equation}
By adding ${\sf P}_\mu(\mc \in I^\omega \land \mc \in L)$ to both sides we obtain:
\begin{equation}
    {\sf P}_\mu(\mc \in I^\omega \land \mc \notin L ) 
    +
    {\sf P}_\mu(\mc \in I^\omega \land \mc \in L )
    = 
    {\sf P}_\mu(\mc \in I^\omega \land \mc \in L ).
\end{equation}
By the law of total probability:
\begin{equation}
    {\sf P}_\mu(\mc \in I^\omega)
    = 
    {\sf P}_\mu(\mc \in I^\omega \land \mc \in L ).
\end{equation}
Dividing both sides by ${\sf P}_\mu(\mc \in I^\omega)$:
\begin{equation}
    1
    = 
    \frac{{\sf P}_\mu(\mc \in I^\omega \land \mc \in L )}{
    {\sf P}_\mu( \mc \in L )
    }.
\end{equation}
Using the definition of conditional expectation we arrive at:
\begin{equation}
    1
    = 
    {\sf P}_\mu(\mc \in L \mid \mc \in I^\omega),
\end{equation}
namely, \cref{eqn:one}.

\section{Proof of \Cref{thm:streett}}

Given the probability transition kernel $P : S \times \Sigma \to [0, 1]$, and the set $I \in \Sigma$ we define a modified transition kernel $P^I : S \times \Sigma \to [0, 1]$ by:
\begin{equation}
    P^I(s, A) = \begin{cases}
        P(s, A) &s \in I\\
        \ind_A(s) &s \notin I 
    \end{cases}.
\end{equation}
Intuitively, transition kernel $P^I$ yields the same behaviour as $P$, except that if at any given time $\Phi_n^I \notin I$ then for all $m \geq n$ we have $\Phi_m^I = \Phi_n^I$ almost surely.
By \Cref{thm:meas-spec}, $P^I$ induces a probability measure and stochastic process $\mc^I$ over specifications ${\sf P}^I_\mu : {\cal F} \to [0,1]$ on the trajectory space $(\Omega, {\cal F})$.

We show that the functions $V_i : S \to \Real_{\geq 0}$ for $i = 1, \ldots, k$ satisfying \cref{eq:dec,eq:inc,eq:noninc} constitute Streett supermartingales \cite[Theorem 2]{DBLP:conf/cav/AbateGR24} proving that the Streett acceptance condition:
\begin{equation}\label{eqn:aug-streett}
    (A_1, B_1 \cup I^{\sf c}) \in \Sigma^2, \ldots,
    (A_k, B_k \cup I^{\sf c}) \in \Sigma^2
\end{equation}
is satisfied almost surely under ${\sf P}^I_\mu$, namely:
\begin{equation}\label{eqn:aug-almost-sure}
        \Prob_\mu^I \left( 
        \bigwedge_{i = 1}^k \sum_{n = 0}^{\infty} \ind_{A_i}(\Phi_n^I) < \infty \lor 
        \sum_{n = 0}^{\infty}\ind_{B_i \cup I^{\sf c}}(\Phi_n^I) = \infty
        \right) = 1. 
    \end{equation}
We argue this by cases, by showing that the functions $V_i : S \to \Real_{\geq 0}$ satisfy the requirements of Streett supermartingales \cite[Theorem 2]{DBLP:conf/cav/AbateGR24} with respect to the acceptance condition \cref{eqn:aug-streett}, for each Streett pair $(A_i, B_i \cup I^{\sf c})$ ranging over $i = 1, \ldots, k$:
\begin{itemize}
    \item Case $s \in A_i \setminus (B_i \cup I^{\sf c})$: we observe that $A_i \setminus (B_i \cup I^{\sf c}) = I \cap ( A_i \setminus B_i )$, in which $V_i : S \to \Real_{\geq 0}$ satisfies $\epsilon$-decrease, by \cref{eq:dec}.
    \item Case $s \in (B_i \cup I^{\sf c})$: we note that if $s \in B_i \cap I$ then the required drift condition follows from \cref{eq:inc}. Otherwise, if $s \notin I$, then since the Markov chain induced by the kernel $P^I$ remains in the same state with probability 1, we have that $P^IV_i(s) = V_i(s) \leq V_i(s) + M$, as required.
    \item Case $s \in S \setminus (A_i \cup B_i \cup I^{\sf c})$: noting that $S \setminus ( A_i \cup B_i \cup I^{\sf c}) = I \setminus (A_i \cup B_i)$, by \cref{eq:noninc} and the fact that $s \in I$ we have $P^IV_i(s) = PV_i(s) \leq V_i(s)$ as required.
\end{itemize}
This establishes, by invoking \cite[Theorem 2]{DBLP:conf/cav/AbateGR24} that \cref{eqn:aug-almost-sure} holds.
Observing that 
\begin{equation}
    \sum_{n = 0}^{\infty}\ind_{B_i \cup I^{\sf c}}(\Phi_n^I) = \infty
\end{equation}
holds if and only if
\begin{equation}
    \sum_{n = 0}^{\infty}\ind_{B_i}(\Phi_n^I) = \infty 
    \lor 
    \sum_{n = 0}^{\infty}\ind_{I^{\sf c}}(\Phi_n^I) = \infty, 
\end{equation}
we may rewrite \cref{eqn:aug-almost-sure} to obtain:
\begin{equation}
    \Prob_\mu^I\left( 
    \bigwedge_{i = 1}^k \sum_{n = 0}^{\infty} \ind_{A_i}(\Phi_n^I) < \infty
    \lor 
    \sum_{n = 0}^{\infty}\ind_{B_i}(\Phi_n^I) = \infty
    \lor
    \sum_{n = 0}^{\infty}\ind_{I^{\sf c}}(\Phi_n^I) = \infty
    \right) = 1, 
\end{equation}
and by propositional logic:
\begin{equation}\label{eqn:Step0}
    \Prob_\mu^I\left(
    \sum_{n = 0}^{\infty}\ind_{I^{\sf c}}  (\Phi_n^I) = \infty
    \lor 
    \bigwedge_{i = 1}^k \sum_{n = 0}^{\infty} \ind_{A_i}(\Phi_n^I) < \infty
    \lor 
    \sum_{n = 0}^{\infty}\ind_{B_i}(\Phi_n^I) = \infty
    \right) = 1.
\end{equation}

We observe that 
\begin{equation}\label{eqn:Step1}
    \Prob_\mu^I \left(
    \mc^I \notin I^\omega 
    \land 
    \sum_{n = 0}^{\infty}\ind_{I^{\sf c}}  (\Phi_n^I) < \infty
    \right) = 0
\end{equation}
since under the transition kernel $P^I$, any trajectory that exits $I$ must necessarily visit $I^{\sf c}$ infinitely many times. Furthermore, for any specification $L_1 \in \mathcal{F}$ we have the following relation between the probability measures ${\sf P}^I_\mu$ and ${\sf P}_\mu$ induced on the trajectory space:
\begin{equation}\label{eqn:preStep2}
    \Prob_\mu^I(\mc^I \in I^\omega \land \mc^I \in L_1) = \Prob_\mu( 
    \mc \in I^\omega \land \mc \in L_1),
\end{equation}
because $P^I$ is equal to $P$ for all states in $I$.

By complementation, this implies that for any $L_2 \in {\cal F}$:
\begin{equation}\label{eqn:Step2}
    \Prob_\mu^I(\mc^I \notin I^\omega \lor \mc^I \in L_2) = \Prob_\mu( 
    \mc \notin I^\omega \lor \mc \in L_2),
\end{equation}

Combining \cref{eqn:Step0,eqn:Step1,eqn:Step2} we conclude 
\begin{equation}
    \Prob_\mu \left(
    \mc \notin I^\omega 
    \lor
    \underbrace{
    \bigwedge_{i = 1}^k \sum_{n = 0}^{\infty} \ind_{A_i}(\Phi_n) < \infty
    \lor 
    \sum_{n = 0}^{\infty}\ind_{B_i}(\Phi_n) = \infty}_{\mc \in L_2}
    \right) = 1.
\end{equation}

By complementation:
\begin{equation}
    \Prob_\mu\left(
    \mc \in I^\omega \land \mc \notin L_2
    \right) = 0.
\end{equation}

Adding $\Prob_\mu ( \mc \in I^\omega \land \mc \in L_2 )$ to both sides,
\begin{equation}
    \Prob_\mu\left(
    \mc \in I^\omega \land \mc \notin L_2
    \right) + \Prob_\mu ( \mc \in I^\omega \land \mc \in L_2 ) = \Prob_\mu ( \mc \in I^\omega \land \mc \in L_2 ).
\end{equation}
Applying the law of total probability:
\begin{equation}
    \Prob_\mu\left(
    \mc \in I^\omega 
    \right) = \Prob_\mu ( \mc \in I^\omega \land \mc \in L_2 ), 
\end{equation}
from which applying the definition of conditional expectation yields the desired \cref{eqn:conditional-1}.

\section{Proof of \Cref{thm:quantitative-streett}}

By \cref{qs:SI-dec,qs:SI-ind}, and \cref{thm:invar} we conclude:
\begin{equation}
    {\sf P}_\mu({\mathbf \Phi} \in I^\omega) \geq 1 - \mu V_0.\label{eqn:lem10-step1}
\end{equation}
By \cref{qs:dec,qs:inc,qs:noninc}, and \cref{thm:streett}, we conclude:
\begin{equation}\label{eqn:lem10-step2}
        \Prob_\mu \left( 
        \bigwedge_{i = 1}^k \sum_{n = 0}^{\infty} \ind_{A_i}(\Phi_n) < \infty \lor 
        \sum_{n = 0}^{\infty}\ind_{B_i}(\Phi_n) = \infty
        \mid \mc \in I^\omega \right) = 1. 
\end{equation}
Applying \cref{thm:decomp} with $p = 1 - \mu V_0$ to \cref{eqn:lem10-step1,eqn:lem10-step2}, we conclude that 
\begin{equation}\label{eqn:lem10-step3}
        \Prob_\mu \left( 
        \bigwedge_{i = 1}^k \sum_{n = 0}^{\infty} \ind_{A_i}(\Phi_n) < \infty \lor 
        \sum_{n = 0}^{\infty}\ind_{B_i}(\Phi_n) = \infty
         \right) \geq 1 - \mu V_0. 
\end{equation}

\section{Case Studies}\label{sec:casestudies}

\subsection{Gambler's Ruin}\label{casestudy:gamblerruin}

\begin{align}
    \Phi_0 &= 10\\
    \Phi_{n+1} &= \begin{cases}
        0 &\Phi_n = 0\\
        \Phi_n + W_n &\Phi_n > 0 
    \end{cases}
\end{align}
where $W_n = 1$ with probability $\frac{51}{100}$, and $W_n = -1$ with probability $\frac{49}{100}$.

\subsection{Gambler's Ruin (control)}\label{casestudy:gamblerruin-control}

\begin{align}
    \Phi_0 &= 10\\
    \Phi_{n+1} &= \begin{cases}
        0 &\Phi_n = 0\\
        \Phi_n + W_n &\Phi_n > 0 
    \end{cases}
\end{align}
where $W_n = 1$ with probability $\frac{1}{2}+\kappa$, and $W_n = -1$ with probability $\frac{1}{2}-\kappa$, $\kappa \in K = \{ \kappa \colon -1/4 \leq \kappa \leq 1/4 \}$.

\subsection{Becoming Rich Once}\label{casestudy:becomingrichonce}

\begin{align}
    \Phi_0 &= 50\\
    \Phi_{n+1} &= \begin{cases}
        0 &\Phi_n = 0\\
        \Phi_n + W_n &\Phi_n > 0 
    \end{cases}
\end{align}
where $W_n = 1$ with probability $\frac{51}{100}$, and $W_n = -1$ with probability $\frac{49}{100}$.

\subsection{Becoming Rich Once (control)}\label{casestudy:becomingrichonce-control}

\begin{align}
    \Phi_0 &= 50\\
    \Phi_{n+1} &= \begin{cases}
        0 &\Phi_n = 0\\
        \Phi_n + W_n &\Phi_n > 0 
    \end{cases}
\end{align}
where $W_n = 1$ with probability $\frac{1}{2} + \kappa$, and $W_n = -1$ with probability $\frac{1}{2} - \kappa$, with $\kappa \in K = \{ \kappa \colon -1/4 \leq \kappa \leq 1/4 \}$.

\subsection{Reactivity 1}\label{casestudy:reactivityOne}

\begin{align}
    \Phi_0 &= 5\\
    \Phi_{n+1} &= \begin{cases}
        \Phi_n + W_n &0 < \Phi_n < 6\\
        \Phi_n - 1 &\Phi_n \leq 0\\
        \Phi_n  &\Phi_n \geq 6
    \end{cases}
\end{align}
where $W_n = 1$ with probability $\frac{1}{2}$, and $W_n = -1$ with probability $\frac{1}{2}$.

\subsection{Reactivity 1 (control)}\label{casestudy:reactivityOne-control}

\begin{align}
    \Phi_0 &= 5\\
    \Phi_{n+1} &= \begin{cases}
        \Phi_n + W_n &0 < \Phi_n < 6\\
        \Phi_n - 1 &\Phi_n \leq 0\\
        \Phi_n  &\Phi_n \geq 6
    \end{cases}
\end{align}
where $W_n = 1$ with probability $\frac{1}{2} + \kappa$, and $W_n = -1$ with probability $\frac{1}{2} - \kappa$, with $\kappa \in K = \{ \kappa \colon -1/4 \leq \kappa \leq 1/4 \}$.

\subsection{Reactivity 2}\label{casestudy:reactivityTwo}

\begin{align}
    \Phi_0 &= 5\\
    \Phi_{n+1} &= \begin{cases}
        0 & \Phi_n = 0\\
        \Phi_n + W_n &1 \leq \Phi_n < 20\\
        \Phi_n + 1 &\Phi_n \geq 20
    \end{cases}
\end{align}
where $W_n = 1$ with probability $\frac{1}{2}$, and $W_n = -1$ with probability $\frac{1}{2}$.

\subsection{Reactivity 2 (control)}\label{casestudy:reactivityTwo-control}

\begin{align}
    \Phi_0 &= 5\\
    \Phi_{n+1} &= \begin{cases}
        0 & \Phi_n = 0\\
        \Phi_n + W_n &1 \leq \Phi_n < 20\\
        \Phi_n + 1 &\Phi_n \geq 20
    \end{cases}
\end{align}
where $W_n = 1$ with probability $\frac{1}{2}$, and $W_n = -1$ with probability $\frac{1}{2}$.
where $W_n = 1$ with probability $\frac{1}{2} + \kappa$, and $W_n = -1$ with probability $\frac{1}{2} - \kappa$, with $\kappa \in K = \{ \kappa \colon -1/4 \leq \kappa \leq 1/4 \}$.

\subsection{RepeatedCoin (control)}\label{casestudy:repeatedCoin-control}

\begin{align}
    \Phi_0 &= 1\\
    \Phi_{n+1} &= \begin{cases}
        0 & \Phi_n = 0\\
        W_n \cdot (\Phi_n + 1) &1 \leq \Phi_n < 20\\
        \Phi_n&\Phi_n \geq 20
    \end{cases}
\end{align}
where $W_n = 1$ with probability $\kappa$ and $W_n = 0$ with probability $1-\kappa$, with $\kappa \in K = \{ \kappa \colon  0 \leq \kappa \leq 1 \}$.

\end{document}